%% file: main-diameter.tex
\documentclass[11pt,english]{article}
\usepackage[utf8]{inputenc}
\usepackage{amsmath}
\usepackage{amsfonts}
\usepackage{amssymb}
\usepackage{amsthm}
\usepackage{thmtools}
\usepackage{xcolor}
\usepackage{array} 
\usepackage{pdflscape} 
\usepackage{algorithmic}

\definecolor{ForestGreen}{rgb}{0.1333,0.5451,0.1333}
\definecolor{DarkRed}{rgb}{0.8,0,0}
\definecolor{Red}{rgb}{1,0,0}
\usepackage[linktocpage=true,
pagebackref=true,colorlinks,
linkcolor=DarkRed,citecolor=ForestGreen,
bookmarks,bookmarksopen,bookmarksnumbered]
{hyperref}
\renewcommand*\backref[1]{\ifx#1\relax \else (cit. on p. #1) \fi} 

\usepackage{cleveref}
\usepackage{geometry}
\usepackage{graphicx}

\newcommand{\N}{\mathbb{N}}
\newcommand{\Z}{\mathbb{Z}}
\newcommand{\I}{\mathbb{I}}
\newcommand{\F}{\mathbb{F}}

\newcommand{\poly}{\operatorname{poly}}
\newcommand{\polylog}{\operatorname{polylog}}

\newcommand{\dist}{\operatorname{dist}}
\newcommand{\diam}{\operatorname{diam}}
\newcommand{\radius}{\operatorname{radius}}

\newcommand{\matrixExponent}{{2.3729}} 

\input{complexities.tex}

\ifdefined\ShowComment
\usepackage{showlabels}

\else

\fi

\declaretheorem[numberwithin=section,refname={Theorem,Theorems},Refname={Theorem,Theorems}]{theorem}

\declaretheorem[numberlike=theorem]{lemma}

\declaretheorem[numberlike=theorem]{corollary}

\declaretheorem[numberlike=theorem]{algorithm}

\declaretheorem[numberlike=theorem, refname={Question,Questions},Refname={Question,Questions},name={Question}]{question}
\declaretheorem[numberlike=theorem, refname={Problem,Problems},Refname={Problem,Problems},name={Problem}]{problem}

\crefformat{footnote}{#2\footnotemark[#1]#3} 

\title{Dynamic Approximate Shortest Paths and Beyond: Subquadratic and Worst-Case Update Time}

\usepackage{authblk}
\author{Jan van den Brand}
\author{Danupon Nanongkai}
\affil{KTH Royal Institute of Technology, Sweden}

\date{} 

\begin{document}

\begin{titlepage}
	\maketitle
	\pagenumbering{roman}
	\ifdefined\ShowComment
	\begin{center}
		{\centering\huge\textcolor{red}{DEBUG VERSION}}
	\end{center}
	\fi
	\input{abstract}

	\newpage 
	
	\setcounter{tocdepth}{2}
	\tableofcontents
	
\end{titlepage}
%
%
%

\pagenumbering{arabic}

\newpage
\input{intro.tex}

\input{overview.tex}

\input{preliminaries.tex}

\input{small_paths.tex}

\input{apsp.tex}

\input{diameter.tex}

\input{open.tex}

\input{acknowledgement}

\bibliographystyle{alpha}
\bibliography{bibliography}

\newpage

\appendix

\input{appendix_worst_case.tex}

\input{appendix_reduction.tex}
\input{appendix_omega.tex}

\input{appendix_hittingset.tex}

\input{appendix_approximate_diameter.tex}

\input{appendix_exact_diameter.tex}

\input{appendix_closeness_centrality.tex}

\end{document}

%% file: complexities.tex
\newcommand{\sankowski}{1.5286}
\newcommand{\preUnwUndAPSP}{2.53}
\newcommand{\updateUnwUndAPSP}{2}
\newcommand{\queryUnwUndAPSP}{2}
\newcommand{\weightUnwUndAPSP}{1.844}

\newcommand{\preWeiDirAPSP}{2.873}
\newcommand{\updateWeiDirAPSP}{2.045}
\newcommand{\queryWeiDirAPSP}{2.045}

\newcommand{\preUnwUndAPSPquery}{2.621}
\newcommand{\updateUnwUndAPSPquery}{1.823}
\newcommand{\queryUnwUndAPSPquery}{0.45}
\newcommand{\sUnwUndAPSPquery}{0.248}
\newcommand{\muUnwUndAPSPquery}{0.202}
\newcommand{\preWeiDirAPSPquery}{2.708}
\newcommand{\updateWeiDirAPSPquery}{1.863}
\newcommand{\queryWeiDirAPSPquery}{0.666}
\newcommand{\sWeiDirAPSPquery}{0.334}
\newcommand{\muWeiDirAPSPquery}{0.316}

\newcommand{\preSubOptSSSP}{\preWeiDirAPSPquery{}}
\newcommand{\updateSubOptSSSP}{\updateWeiDirAPSPquery{}}
\newcommand{\querySubOptSSSP}{1.666} 

\newcommand{\preSubOptUnwUndSSSP}{\preUnwUndAPSPquery{}}
\newcommand{\updateSubOptUnwUndSSSP}{\updateUnwUndAPSPquery{}}
\newcommand{\querySubOptUnwUndSSSP}{1.45} 

\newcommand{\preSSSP}{2.621}
\newcommand{\updateSSSP}{1.823}
\newcommand{\sSSSP}{0.248}
\newcommand{\muSSSP}{0.202}

\newcommand{\preRadius}{2.624}
\newcommand{\updateRadius}{1.779}
\newcommand{\sRadius}{0.25}
\newcommand{\muRadius}{0.4}

\newcommand{\updateExactDiameter}{2.3452}
\newcommand{\sExactDiameter}{0.6548}
\newcommand{\muExactDiameter}{0.3097}

%% file: abstract.tex

\begin{abstract}
	Consider the following distance query for an $n$-node graph $G$ undergoing edge insertions and deletions: given two sets of nodes $I$ and $J$, return the distances between every pair of nodes in $I\times J$.  
	This query is rather general and captures several versions of the dynamic shortest paths problem. 
	In this paper, we develop an efficient $(1+\epsilon)$-approximation algorithm for this query using fast matrix multiplication. Our algorithm leads to answers for some open problems for Single-Source and All-Pairs Shortest Paths (SSSP and APSP), as well as for Diameter, Radius, and Eccentricities. Below are some highlights. Note that all our algorithms guarantee {\em worst-case} update time and  are randomized (Monte Carlo), but do not need the oblivious adversary assumption.

	{\em Subquadratic update time for SSSP, Diameter, Centralities, ect.:} 
	When we want to maintain distances from a single node explicitly (without queries), 
	a fundamental question is to beat trivially calling Dijkstra's static algorithm after each update, 
	taking $\Theta(n^2)$ update time on dense graphs. 
	A better time complexity was not known even with amortization. 
	It was known to be improbable for {\em exact} algorithms 
	and for {\em combinatorial} any-approximation algorithms
	to polynomially beat the $\Omega(n^2)$ bound (under some conjectures) 
	[Roditty, Zwick, ESA'04; Abboud, V. Williams, FOCS'14].\footnote{%
	The conditional lower bounds of [Roditty, Zwick, ESA'04; Abboud, V. Williams, FOCS'14] 
	hold for algorithms with $O(n^{3-\delta})$ preprocessing time for some constant $\delta>0$. 
	Our preprocessing time is also in this form. 
	``Combinatorial algorithms'' is a vague term referring to algorithms that do not use fast matrix multiplication.}	
	%
	Our algorithm with $I=\{s\}$ and $J=V(G)$ implies a $(1+\epsilon)$-approximation algorithm for this, guaranteeing  $\tilde O(n^{\updateSSSP}/\epsilon^2)$ worst-case update time for directed graphs with positive real weights in $[1, W]$.\footnote{\label{foot:notation}{\bf Notations:} Throughout, $n$ and $m$ denote the number of nodes and edges respectively. Our focus is on dense graphs with $m=\Theta(n^2)$ edges. Let $V(G)$ denote the set of nodes in a graph $G$. Unless specified otherwise, ``weighted graphs'' refer to graphs with positive real edge weights. The $\tilde O$ notation hides $\poly\log n$, $\poly\log (1/\epsilon)$ (for $(1+\epsilon)$-approximation algorithms), and during the introduction a single $\log W$ factor.}
	With ideas from [Roditty, V. Williams, STOC'13], we also obtain the first subquadratic worst-case update time for $(5/3+\epsilon)$-approximating the eccentricities and $(1.5+\epsilon)$-approximating the diameter and radius  for unweighted graphs (with small additive errors). We also obtain the first  subquadratic worst-case update time  for $(1+\epsilon)$-approximating the closeness centralities for undirected unweighted graphs. 
	
	{\em Worst-case update time for APSP:} When we want to maintain distances between all-pairs of nodes explicitly, the $\tilde O(n^2)$ {\em amortized} update time by Demetrescu and Italiano [STOC'03] already matches the trivial  $\Omega(n^2)$ lower bound. A fundamental question is whether it can be made {\em worst-case}. The state-of-the-art algorithm takes  $\tilde O(n^{2+2/3})$ worst-case update time to maintain the distances exactly  [Abraham, Chechik, Krinninger, SODA'17; Thorup STOC'05]. When it comes to $(1+\epsilon)$ approximation, this bound is still higher than calling the  $\tilde O(n^\omega/\epsilon)$-time static algorithm of Zwick [FOCS'98], where $\omega\approx 2.373$. 
	Our algorithm with $I=J=V(G)$ implies nearly tight bounds for this, namely  $\tilde O(n^2/\epsilon^{1+\omega})$ for undirected unweighted graphs and $\tilde O(n^{2.045}/\epsilon^2)$ for directed graphs with positive real weights. 
	Besides this, we also obtain the first dynamic APSP algorithm with subquadratic update time and sublinear query time. 
\end{abstract}

%% file: intro.tex
\section{Introduction }\label{sec:intro}
Dynamic graph algorithms generally concern maintaining properties of a graph under a sequence of updates, typically in the form of an edge insertion, deletion or weight update. Among basic primitives extensively studied are various distance information; e.g., the all-pairs shortest paths (APSP), the single-source shortest paths (SSSP), and the $st$-shortest path (st-SP) concern the distances between all-pairs of nodes, from a single node to every node, and between a pair of nodes, respectively.\footnote{Note that some works also considered returning the shortest paths, not just the distances, and a {\em node update} where all edges incident to the same nodes are updated together. This paper does not consider this.}
These problems have been studied in settings where distances can be {\em queried} 
(e.g. \cite{HenzingerK95,HenzingerKN13,Bernstein09,RodittyZ12,Sankowski04,Sankowski05}) 
or are {\em explicitly} maintained 
(e.g. \cite{DemetrescuI04,Thorup05,AbrahamCK17}).
In this paper, we study the problem that captures many aforementioned problems as special cases. In this problem, one can query the distances between any two sets of nodes $I$ and $J$, as follows. 

\begin{problem}[Dynamic batch-query distances]
\label{prob:batch-query}
An algorithm for this problem supports the following operations.
	\begin{itemize}
		\item {\sc Preprocess}($G$): Process an input $n$-node graph $G$ with positive real edge weights from $[1,W]$.
		\item {\sc Update}($(u, v), w$):  Update the weight of edge $(u, v)$ to $w\in [1, W]\cup \{\infty\}$.\footnote{Setting edge weight to $\infty$ is equivalent to deleting an edge.}
		\item {\sc Query}($I, J$):  Given sets of nodes $I$ and $J$, return the distance from $i$ to $j$, denoted by $\dist(i,j)$, for each $(i,j)\in I\times J$.
	\end{itemize}
\end{problem}

For example, the explicit dynamic SSSP where we maintain the distances from a pre-specified node $s$ to every node after changing the weight of an edge 
is a special case of the above problem where we fix $I=\{s\}$ and $J=V(G)$ (where $V(G)$ is the set of all nodes in $G$), and the query is made after every update. 

Sometimes, we call the query in \Cref{prob:batch-query} the {\em batch query} to distinguish it from a typical query of the distance between two nodes. 
We refer to the case where $W=1$ as the {\em unweighted case} (setting an edge weight to $1$ and $\infty$ corresponds to inserting and deleting an edge). To keep things short, we use {\em weighted graphs} to refer to graphs with positive real edge weights in $[1, W]$ (like in \Cref{prob:batch-query}) throughout.

The performance of algorithms for  \Cref{prob:batch-query} is measured by {\em preprocessing time}, {\em update time} and {\em query time}.
The update time can be categorized into two types: A more desirable one is the {\em worst-case} update time  which holds for every single update. This is to contrast with an {\em amortized} update time which holds ``on average''.\footnote{More precisely, for any $t$, an algorithm is said to have an amortized update time of $t$ if, for any $k$, the total time it spends to process the first $k$ updates is at most $kt$.}
Our focus is on the {\em worst-case} update time and {\em dense} graphs with $m=\Theta(n^2)$ edges. 

In this paper, we present a fast {\em $(1+\epsilon)$-approximation} algorithm for  \Cref{prob:batch-query}. By  $(1+\epsilon)$-approximation, we mean that it answers {\sc Query}$(I, J)$ with $d'(i, j)\in [\dist(i,j), (1+\epsilon)\dist(i, j)]$ for every $(i, j)\in I\times J$. 
We state our algorithm's performance below for completeness but recommend the reader to skip on the first read. 
It might also be helpful to focus the bounds below for unweighted undirected graphs, where $\omega=2$, $\varepsilon=0.01$, and $s\leq 1/4$. In this case, the update and query time complexities in \Cref{thm:intro-main} become $\tilde{O}(n^{2-s})$ and $\tilde{O}(|I||J|+|I|n^{2s}+|J|n^{2s})$, which can be traded-off using the parameter $s$.

\begin{theorem}[Main General Result]
\label{thm:intro-main}
For any $0<s<1$, there exists a randomized (Monte Carlo) $(1+\epsilon)$-approximation algorithm for  \Cref{prob:batch-query} whose preprocessing time, {\em worst-case} update time, and query time on directed weighted graphs (respectively undirected unweighted graphs) are 
	\ifdefined\FOCSversion
	\begin{itemize}
		\item Preprocessing: $\tilde{O}(n^{\omega+s} / \varepsilon)$ (respectively  $\tilde{O}(n^{\omega+s})$),
		\item Update:\\
		$\tilde{O}((n^{\sankowski+s}  
		+ n^{\omega(1,1,1-s)-1+2s} 
		+ n^{\omega(1,1-s,1-s)})
		/\varepsilon^2)$
		(respectively  $\tilde{O}((n^{\sankowski+s} 
		+ n^{\omega(1,1,\mu+s)-\mu}
		+ n^{\omega(1,\mu+s,1-s)} 
		+n^{(1-s)\omega})/ \varepsilon^{\omega+1}$), and
		\item Query:\\
		$\tilde{O}(n^{\omega(\delta_1,1-s,\delta_2)} / \varepsilon^2)$ 
		(respectively  $\tilde{O}(n^{\omega(\delta_1,s+\mu,\delta_2)} / \varepsilon)$), 
		for $|I| = n^{\delta_1}$ and $|J| = n^{\delta_2}$).
	\end{itemize} 
	\else
\begin{itemize}
		\item Preprocessing: $\tilde{O}(n^{\omega+s} / \varepsilon)$ (respectively  $\tilde{O}(n^{\omega+s})$),
		\item Update: 
		$\tilde{O}(n^{\sankowski+s} / \varepsilon 
		+ n^{\omega(1,1,1-s)-1+2s} / \varepsilon^2 
		+ n^{\omega(1,1-s,1-s)}/\varepsilon^2)$\\
		(respectively  $\tilde{O}(n^{\sankowski+s} + n^{\omega(1,1,\mu+s)-\mu} / \varepsilon + n^{\omega(1,\mu+s,1-s)} / \varepsilon^2 +n^{(1-s)\omega}/ \varepsilon^{\omega+1}$), and
		\item Query: $\tilde{O}(n^{\omega(\delta_1,1-s,\delta_2)} / \varepsilon^2)$ (respectively  $\tilde{O}(n^{\omega(\delta_1,s+\mu,\delta_2)} / \varepsilon)$), for $|I| = n^{\delta_1}$ and $|J| = n^{\delta_2}$).
	\end{itemize} 
	\fi
\end{theorem}

The $\tilde O$ notation hides $\poly\log(n)$ and $\poly\log(1/\epsilon)$. To simplify our discussions, we hide the $\log W$ term in this section; we emphasize that our algorithms have only a single $\log W$ factor in their complexities.
Here $\omega$ is the matrix multiplication exponent, 
i.e. multiplying two $n \times n$ matrices requires $O(n^\omega)$ time 
(the current best bound is $\omega \le \matrixExponent$ \cite{Gall14a,Williams12}) 
and $\omega(a, b, c)$ is the exponent for multiplying 
an $n^a \times n^b$ matrix with an $n^b \times n^c$ matrix.
\ifdefined\FOCSversion
For bounds on $\omega(a,b,c)$ we refer to the appendix of the full version of this paper \cite{BrandN19}.
\else
For bounds on $\omega(a,b,c)$ see \Cref{app:omega}.
\fi
The claimed algorithm, as well as other algorithms that follow, 
guarantee worst-case update and query times and are randomized in that 
they return the results within the guaranteed approximation ratio with high probability%
\footnote{With high probability (w.h.p.) means with probability at least $1-1/n^c$ for any constant $c>1$.}. 
Note that unlike typical randomized dynamic algorithms, 
our algorithms do {\em not} need the oblivious adversary assumption; 
i.e. an edge update can depend on the algorithms' prior outputs.

\subsection{Consequences}
Our result leads to improved algorithms for several variants of dynamic distance maintenance. This includes maintaining the {\em diameter}, the {\em radius}, and  the  {\em eccentricities}.
It answers some key questions in the studies of dynamic shortest paths. Below are some of these questions (more explanations will follow). 

\begin{question}[Beating static algorithms]
\label{question:break_static}
Can we beat trivially calling static algorithms after every update? 
In particular, can we (i) explicitly maintain the SSSP in {\em subquadratic} update time
(even with amortization), and (ii) explicitly and $(1+\epsilon)$-approximately maintain the APSP faster than $\tilde O(n^\omega)$ {\em worst-case} update time? 
\end{question}

\begin{question}[De-amortization]
\label{question:worst-case}
Can we achieve a worst-case update time comparable to the best known amortized update time?
In particular, can we explicitly, and perhaps approximately, maintain the APSP in $\tilde O(n^2)$ worst-case update time?
\end{question}

\paragraph{SSSP}
Beating static algorithms (\Cref{question:break_static}) is the first step in tackling any dynamic problems. It has been achieved for a great number of problems, but unfortunately not for a basic problem like SSSP where we want to explicitly maintain distances between a pre-specified node $t$ and all other nodes:
There was no exact or $(1+\epsilon)$-approximation algorithm
that beats calling Dijkstra's $O(m+n\log n)$-time static algorithm after every update, causing $\Theta(n^2)$ update time on dense graphs.
This is in contrast to the {\em partially-dynamic} setting, where a $(1+\epsilon)$-approximation algorithm with $\tilde O(n)$ amortized update time was known essentially since 1981 \cite{EvenS81} (see, e.g.  \cite{HenzingerKN18-DecrementalSSSP,HenzingerKN14,HenzingerKN-SODA14,HenzingerKN-ICALP13,BernsteinC16,Bernstein17,BernsteinR11}
for some recent improvements).\footnote{Recall that the partially-dynamic setting is when edge weighs can be only increased or only decreased. \cite{EvenS81} originally presented an exact algorithm with $O(n)$ amortized update time for unweighted undirected graphs. It was observed later than this can be easily extended to $\tilde O(n)$ amortized update time for weighted directed graphs.}
The question for the fully-dynamic setting was raised in, e.g., \cite{DemetrescuI04}.

It was known that polynomially beating this bound is {\em improbable} for exact algorithms: there is no {\em exact} dynamic SSSP algorithm with $O(n^{3-\delta})$ preprocessing time and $O(n^{2-\delta})$ amortized update time for any constant $\delta>0$, assuming the so-called APSP conjecture \cite{RodittyZ11,AbboudW14}.\footnote{The result of  \cite{RodittyZ11,AbboudW14} does not rule out an exact algorithm with higher preprocessing time and $O(n^{2-\delta})$ update time. Finding such algorithm remains a major open problem.}
A natural question is whether this can be achieved with a {\em $(1+\epsilon)$-approximation} algorithm.
\Cref{thm:intro-main} shows that this is the case:
it implies $\tilde O(n^{\updateSubOptSSSP})$ worst-case update time for SSSP (after the $\tilde O(n^{\preSubOptSSSP})$-time preprocessing) for directed weighted graphs. 
This is simply because, for $|I|=1$, $|J|=n$, and an appropriate choice of $s$,
the preprocessing time, update time, and query time in \Cref{thm:intro-main} become
$\tilde O(n^{\preSubOptSSSP})$,
$\tilde O(n^{\updateSubOptSSSP})$,
$\tilde O(n^{\querySubOptSSSP})$,
respectively for directed weighted graphs.
For unweighted directed graphs, the bounds become
$O(n^{\preSubOptUnwUndSSSP})$, 
$\tilde O(n^{\updateSubOptUnwUndSSSP})$, 
and $\tilde O(n^{\querySubOptUnwUndSSSP})$, 
implying slightly lower time complexities.
By adjusting a small part of the proof of \Cref{thm:intro-main}, we get slightly better bounds:

\begin{theorem}\label{thm:intro-SSSP-subquadratic}
\emph{SSSP in subquadratic update time; Details in \Cref{thm:SSSP}):}
There is a $(1+\epsilon)$-approximation algorithm for maintaining SSSP explicitly for directed weighted graphs in  $\tilde O(n^{\updateSSSP}/ \varepsilon^{2})$ worst-case update time after the $\tilde O(n^{\preSSSP})$-time preprocessing. 
\end{theorem}

In the best future scenario, when $\omega=2$, the update time in \Cref{thm:intro-SSSP-subquadratic} would become $\tilde O(n^{1.75}/\varepsilon^2)$. 
Prior to our work, the only method to beat the $\Theta(n^2)$ update time was to run a static algorithm (e.g. Dijkstra's) on top of a {\em dynamic sparse spanner}, giving approximation factors of three or more on undirected graphs.  For example, Bernstein, Forster, and Henzinger \cite{BernsteinFH19} can maintain a $(2k-1)$-spanner of $\tilde O(n^{1+1/k})$ edges in $O(1)^k\log^3(n)$ worst-case update time, for any constant $k\geq 1$ and for undirected graphs. This allows us to $3$-approximate SSSP on undirected graphs in $\tilde O(n^{1.5})$ worst-case update time. 
Since spanners work only for {\em undirected} graphs and cause the distances to increase by a multiplicative factor of at least three,  its use is fundamentally limited to such graphs and approximation guarantee. Our algorithm avoids spanners completely and works for directed graphs with much lower approximation ratio.

\paragraph{APSP}  

The algorithm with $\tilde O(n^{2.75})$ worst-case update time of Thorup \cite{Thorup05} was among the first that addressed the issue of worst-case update time (\Cref{question:worst-case}). 
Despite much recent effort and progress on this issue\footnote{\label{foot:worst-case literature}See, e.g., \cite{AbrahamCK17,Thorup05,NanongkaiSW17,NanongkaiS17,Wulff-Nilsen17,BhattacharyaHN17,KapronKM13,BernsteinFH19,CharikarS18,KopelowitzKPS14,ArarCCSW18,BodwinK16}.}, the only improvement over Thorup's bound was the $\tilde O(n^{2+2/3})$ worst-case update time by Abraham, Chechik, and Krinninger \cite{AbrahamCK17}. 
This bound holds for directed weighted graphs and can be improved to $\tilde O(n^{2.5})$ on directed {\em unweighted} graphs. (Update: Recently, after our publication, Wulff-Nilsen and Probst analyzed worst-case APSP in the deterministic and Las Vegas setting \cite{WNilsenP20}.)

In fact, the above results are all for maintaining distances {\em exactly}. When it comes to maintaining {\em $(1+\epsilon)$-approximate} distances, they are not better than trivially running a static algorithm after every update. On dense directed weighted graphs, this takes  $\tilde O(n^\omega/\epsilon)$ worst-case update time due to Zwick's algorithm \cite{Zwick02}. (Again, $\omega\approx \matrixExponent$.)  In other words, there was no $(1+\epsilon)$-approximation algorithm that beats static algorithms (\Cref{question:break_static}). 
\Cref{thm:intro-main} implies algorithms that do not only break the static $\tilde O(n^\omega/\epsilon)$ bound, but are also nearly tight:

\begin{theorem}[APSP in almost-quadratic worst-case update time; 
Details in \Cref{thm:weightedAPSP}]
\label{thm:intro-APSP-weighted}
There is a $(1+\epsilon)$-approximation algorithm 
for maintaining all-pairs-distances explicitly in
(i) $\tilde O(n^{\updateUnwUndAPSP}/ \varepsilon^{\omega+1})$
worst-case update time after the $\tilde O(n^{\preUnwUndAPSP})$-time preprocessing
for undirected unweighted graphs,
and (ii) $\tilde O(n^{\updateWeiDirAPSP}/ \varepsilon^{2})$,  
worst-case update time after the $O(n^{\preWeiDirAPSP})$-time preprocessing 
for directed weighted graphs.
\end{theorem}

This is simply because for $I=J=V(G)$ and an appropriate choice of $s$,  
the preprocessing time, worst-case update time, and query times in \Cref{thm:intro-main} 
become $O(n^{\preUnwUndAPSP})$,
$\tilde O(n^{\updateUnwUndAPSP})$, 
and $\tilde O(\queryUnwUndAPSP)$ 
for undirected unweighted graphs,
and $O(n^{\preWeiDirAPSP})$, 
$\tilde O(n^{\updateWeiDirAPSP})$ 
and $\tilde O(n^{\queryWeiDirAPSP})$ 
for directed weighted graphs.
Prior to our algorithms, the only way to beat Zwick's static algorithm was via dynamic spanners; e.g., the aforementioned algorithm of Bernstein~et~al.~\cite{BernsteinFH19} implies a $5$-approximation algorithm with $\tilde O(n^{2+1/3})$ worst-case update time and an $O(1/\epsilon)$-approximation algorithm with $\tilde O(n^{2+\epsilon})$ worst-case update time.
As mentioned earlier, this approach is fundamentally limited to undirected graphs and rather large approximation ratios.

Note that while previous algorithms \cite{Thorup05,AbrahamCK17}
can also return the shortest path connecting two nodes 
in time proportional to the length of the path, 
our algorithms only maintain the distances.
Also, previous algorithms can handle a more general update 
where the weights of all edges incident to the same node are updated at once. 
Our algorithms only handle the standard edge updates. 
Due to the so-called ``Johnson transformation''\cite{Johnson77,CormenLRS09},
previous algorithms can also handle negative edge weights 
when there are no negative cycles. 
Since this transformation applies only for exact distance computation, 
it does not apply to our algorithms.

For a version of APSP where we can make a query for a distance between two nodes, \Cref{thm:intro-main} implies a $(1+\epsilon)$-approximation algorithm with {\em subquadratic} update and {\em sublinear} query time complexities (both are in the worst case).

\begin{theorem}[APSP in subquadratic update and sublinear query time; 
Details in \Cref{thm:weightedAPSPQuery}]
\label{thm:intro-APPS-queries}
There is a $(1+\epsilon)$-approximation algorithm for maintaining APSP
with $\tilde O(n^{\updateWeiDirAPSPquery}/ \varepsilon^{2})$ worst-case update time 
and $\tilde O(n^{\queryWeiDirAPSPquery}/ \varepsilon^{2})$ worst-case query time 
after the $\tilde O(n^{\preWeiDirAPSPquery})$-time preprocessing 
for directed weighted graphs,
and (ii) $\tilde O(n^{\updateUnwUndAPSPquery} / \varepsilon^{\omega+1})$ worst-case update time 
and $\tilde O(n^{\queryUnwUndAPSPquery}/ \varepsilon^{\omega+1})$ worst-case query time 
after the $ O(n^{\preUnwUndAPSPquery})$-time preprocessing 
for undirected unweighted graphs.
\end{theorem}

The only previously known algorithm with subquadratic update and query time complexities (but not with a sublinear query time) was by Sankowski~\cite{Sankowski05}. It outputs exact distances for directed unweighted graphs. Our algorithm works on weighted graphs with much lower query time, but only returns approximate distances.

\paragraph{Diameter, Radius, and Eccentricities} 
The eccentricities (the eccentricity of a node $v$ is the largest distance from $v$ to any another node), the diameter (the maximum  over all eccentricities) and the radius (the minimum over all eccentricities) can be easily maintained in $\tilde O(n^2)$ amortized update time using Demetrescu and Italiano's dynamic APSP algorithm \cite{DemetrescuI04}. 
An important challenge here is to {\em break the $\tilde O(n^2)$  bound}. This  captures a fundamental question of whether we really need APSP to maintain less informative measurements like the diameter.
It was known that algorithms with $n^{2-\delta}$ amortized update time might not exist for  $(1.5-\epsilon)$-approximating the diameter, $(1.5-\epsilon)$-approximating the radius, and $(5/3-\epsilon)$-approximating the eccentricities for any constants $\delta,\epsilon>0$ (assuming either the Strong Exponential Time Hypothesis (SETH) or a version of the Hitting Set Hypothesis) \cite{AnconaHRVW18}.\footnote{\cite{AnconaHRVW18} also ruled out algorithms with $m^{1-\delta}$ update time that is $(2-\epsilon)$-approximation on undirected unweighted graphs (under SETH),  and finite approximation on directed unweighted graphs (under the k-Cycle Hypothesis), strengthening \cite{HenzingerKNS15}.} In other words, likely we cannot break the $O(n^2)$ bound with such approximation guarantees.

In this paper, we show that with slightly higher approximation guarantees we can break the $O(n^2)$ bound:
by essentially simulating the static algorithm of Roditty and V. Williams \cite{RodittyW13} (with some small adjustments), using the algorithm in \Cref{thm:intro-main} to compute distances when needed, we obtain algorithms with a subquadratic worst-case update time that nearly $(1.5+\epsilon)$-approximate the diameter and radius, and nearly $(5/3+\epsilon)$-approximate the eccentricities for dynamic graphs.  By ``nearly'', we mean that there are some additive errors smaller than one. The results, when we slightly adjust the proof of \Cref{thm:intro-main} to get some slight improvements, are as follows.\footnote{
	Without the adjustment, our algorithms guarantee the following:
	nearly ($1.5+\epsilon$) approximation factor
	and $O(n^{\updateWeiDirAPSPquery} / \varepsilon^{\omega+1})$ time for Diameter,
	nearly ($1.5+\epsilon$) approximation factor 
	and $O(n^{\updateUnwUndAPSPquery} / \varepsilon^{\omega+1})$ time for Radius, 
	and nearly ($5/3+\epsilon$) approximation factor 
	and $O(n^{\updateUnwUndAPSPquery} / \varepsilon^{\omega+1})$ time for Eccentricities.
	Note that the adjustment does not improve the update time for Eccentricities.}

\begin{theorem}[Approximating Diameter, Radius, and Eccentricities; 
Details in \Cref{thm:1.5approxDiam}]
We write $\diam(G)$,  $\radius(G)$ for the diameter and radius of graph $G$ respectively, 
and let $ecc(v, G)$ be the eccentricity of node $v$ in $G$.  
	There exist algorithms that can maintain the following values with the following time complexities for a dynamic graph $G$.
	
	\begin{enumerate}
		\item $	\tilde{D} \in \left[\left(\frac{2}{3}-\varepsilon\right) \diam(G) - 1/3, (1+\varepsilon) \diam(G)\right]$
		in $\tilde O(n^{\updateRadius}/\varepsilon^{1+\omega})$ worst-case update time
		after the $O(n^{\preRadius})$-time preprocessing, 
		\item $\tilde R\in \left[\radius(G) / (1+\varepsilon), (1.5+\varepsilon) \radius(G) + 2/3\right]$ 	
		in $\tilde O(n^{\updateRadius}/\varepsilon^{1+\omega})$ worst-case update time 
		after the $O(n^{\preRadius})$-time preprocessing, and 
		\item	$ \widetilde{ecc}(v) \in \left[\left(\frac{3}{5}-\varepsilon\right) ecc(v,G) - 4/7 , ecc(v,G)\right]$ 
		 for all nodes $v$ in $\tilde O(n^{\updateUnwUndAPSPquery}/\varepsilon^{1+\omega})$ worst-case update time 
		 after the $O(n^{\preUnwUndAPSPquery})$-time preprocessing.
	\end{enumerate}
	The algorithm for Diameter works for directed unweighted graphs, while the others work for undirected unweighted graphs.
\end{theorem}

Prior to our algorithms, one can guarantee similar approximation ratios by running the static algorithms after each update (e.g.  the nearly $1.5$-approximation $\tilde O(m^{3/2})$-time algorithms of \cite{RodittyW13,ChechikLRSTW14} for Diameter and Radius, and the nearly $(5/3)$-approximation $\tilde O(m^{3/2})$-time algorithms of \cite{RodittyW13,ChechikLRSTW14} for Eccentricities).
(See, e.g., \cite{AingworthCIM99,RodittyW13,ChechikLRSTW14,CairoGR16,BackursRSWW18} for results in the static setting.)
Obviously, even a static linear-time algorithm cannot  break the  $\tilde O(n^2)$ bound for dense graphs. 
The only prior method to break the $\tilde O(n^2)$ bound was to run static algorithms 
that are subquadratic-time on sparse graphs on top of a dynamic sparse spanner; 
e.g., the aforementioned spanner algorithm of Bernstein~et~al.~ \cite{BernsteinFH19} 
allows us to nearly-$7.5$-approximate the diameter on directed unweighted graphs
in $\tilde O(n^{1+1/2+1/3})$ worst-case update time 
(using the static $1.5$-approximation $\tilde O(m\sqrt{n})$-time algorithm of \cite{RodittyW13} for Diameter). 
As mentioned earlier, this method is limited to large approximation ratio and undirected graphs. 

Recently, Ancona~et~al. presented {\em partially-dynamic} algorithms 
with approximation guarantees similar to us, 
e.g. a nearly-$(1.5+\epsilon)$-approximation decremental algorithm 
with $m^{1+o(1/\epsilon)}\sqrt{n}/\epsilon^2$ expected total update time 
for unweighted undirected graphs~\cite[Cor.~1.1]{AnconaHRVW18}. 
(A slower algorithm was presented in \cite{ChoudharyG18}.)
Both Ancona~et~al.'s and our algorithms 
essentially simulate the algorithms of Roditty and V, Williams \cite{RodittyW13} 
(with small adjustments to the algorithms and analyses). 
The main difference is that our algorithms rely on our new result in \Cref{thm:intro-main}, 
while Ancona~et~al.'s algorithms rely on the recent developments 
on partially-dynamic shortest paths 
(e.g. \cite{HenzingerKN18-DecrementalSSSP,HenzingerKN13}).

In addition to the above, we can also maintain the diameter {\em exactly} for directed unweighted  graphs (or with small positive integer weights bounded by $W$). 
It is the first that improves trivially running the static exact algorithms by
Shoshan and Zwick\cite{ShoshanZ99} or Cygan, Gabow  and Sankowski \cite{CyganGS15}.
These algorithms take $\tilde O(Wn^\omega)\approx \tilde O(Wn^{2.3729})$ worst-case update time, 
and ours improves this to $\tilde O(Wn^{\updateExactDiameter})$.
\ifdefined\FOCSversion
The algorithm can be found in the full version.
\else 
The algorithm can be found in \Cref{app:exactDiameter}.
\fi

Further, we give the first subquadratic dynamic algorithm 
for $(1+\varepsilon)$-approximating all-closeness-centralities 
with $O(n^{\updateUnwUndAPSPquery} / \varepsilon^{\omega+1})$ update time 
\ifdefined\FOCSversion
the full version.
\else 
in \Cref{app:closenessCentrality}.
\fi
The closeness centrality of a node $v$ is the inverse of the average distance from $v$.
The technique can be used to obtain a static $(1+\varepsilon)$-approximate algorithm 
with $\tilde{O}(mn^{2/3} / \varepsilon^2)$ time algorithm, 
which improves upon the $\tilde{O}(m \cdot \diam(G)^2 / \varepsilon^2)$ algorithm 
from \cite{EppsteinW04} for large diameter graphs.

\paragraph{Techniques} At the heart of our result is a combination of the standard hitting argument and an algorithm that exploits fast matrix multiplication to quickly answer queries about approximate {\em bounded-hop} distances between two sets of nodes. An exact counterpart of this algorithm was already known due to Sankowski \cite{Sankowski05}, but the query time of Sankowski's algorithm is too slow for our purpose. Our approximation algorithm's query time is faster, but its update time is slower. Interestingly, we run Sankowski's algorithm in parallel with our algorithm, since our algorithm needs some information from it during the updates. Our algorithm needs to make a small number of queries to Sankowski's algorithm, thus does not suffer from its slow query time. We explain this more in \Cref{sec:overview}. 

We note that our algorithms heavily rely on fast matrix multiplication algorithms. 
This is known to be necessary even for $st$-Reachability 
and $(1+\epsilon)$-approximating $st$-distance on undirected unweighted graphs \cite{AbboudW14}.\footnote{%
Triangle detection can be reduced to $(6/5-\varepsilon)$-approximate $st$-distance and $(3/2-\varepsilon)$-approximate SSSP.
}
It is an intriguing question, however, whether fast matrix multiplication is necessary for other problems (see \Cref{sec:open}).
Also note that fast matrix multiplication has been used in many previous shortest paths data-structures (e.g. \cite{Sankowski04,Sankowski05,DemetrescuI05,Williams11,GrandoniW12,WeimannY13,BrandNS18,BrandS18}).

%% file: overview.tex
\section{Technical Overview}\label{sec:overview}

In this section we outline how to obtain \Cref{thm:intro-main}. 
For simplicity we will only consider the case of directed, unweighted graphs.
The algorithm outlined here can easily be extended to small integer weights,
which then allows an extension to real weights via weight rounding 
 -- a technique previously used in \cite{Bernstein16,Madry10,RaghavanT85,Zwick02,Nanongkai14}.

We will outline the proof of the following Theorem. 
It is equivalent to \Cref{thm:intro-main}, 
except that it only supports unweighted graphs.

\begin{theorem}[\Cref{thm:intro-main} restricted to directed, unweighted graphs]
\label{thm:overview:main}
For any $0<s<1$, there exists a randomized (Monte Carlo) 
$(1+\epsilon)$-approximation algorithm for \Cref{prob:batch-query} 
whose preprocessing time, 
{\em worst-case} update time, 
and query time 
on directed unweighted graphs are 
	\begin{itemize}
		\item Preprocessing: 
			$\tilde{O}(n^{\omega+s})$,
		\item Update: \\
			$\tilde{O}(n^{\sankowski+s} + n^{\omega(1,1,1-s)-1+2s} / \varepsilon +n^{\omega(1,1-s,1-s)}/\varepsilon)$, 
			and
		\item Query: 
			$\tilde{O}(n^{\omega(\delta_1,1-s,\delta_2)})$, 
			for $|I| = n^{\delta_1}$ and $|J| = n^{\delta_2}$).
	\end{itemize} 
\end{theorem}

The outline of this section is as follows:
We first give the high-level idea 
of how to reduce \Cref{thm:overview:main} 
to the algebraic problem of maintaining the inverse 
of some polynomial matrix modulo $X^h$ for some $h > 0$.
This reduction is outlined in the first subsection \ref{sec:overview:basics}
and uses common techniques used in many other algorithms 
(i.e. weight rounding, hitting sets, and inverse of polynomial matrices).
Readers familiar with dynamic algebraic algorithms 
for maintaining graph distances
can skip ahead to the next subsection \ref{sec:overview:new}.
In \cref{sec:overview:new} we highlight the difference and new techniques 
that allow us to maintain approximate distances quickly,
whereas previous algebraic data-structures could only maintain exact distances.
We also give a high-level description of our new data-structure 
and explain how we are able to break 
the long-standing $\Omega(n^2)$-bottleneck 
for problems such as single-source distances.

\subsection{Basic Tools}
\label{sec:overview:basics}

In this subsection we outline how to reduce \Cref{thm:overview:main} to some algebraic problem.
For that we first reduce the problem to maintaining only $h$-hops distances for some $h > 0$, 
and then reduce the problem further to the algebraic problem of maintaining the inverse 
of some polynomial matrix modulo $X^h$.

\paragraph{Restriction to short hops}

In order to create an algorithm as stated in \Cref{thm:overview:main}, 
it is enough to construct an algorithm that only maintains $h$-hops distances.
Specifically, it is enough to prove the following Lemma. 

\begin{lemma}[Proven as \Cref{thm:dynamicSmallApproximateDistancesQuery} in \Cref{sec:smallHop}]
\label{lem:overview:short hop}
Let $G$ be an unweighted graph with $n$ nodes.
Then for any $0 \le \mu$, $0 \le s \le 1$ and $\varepsilon > 0$
there exists a dynamic algorithm that maintains
$(1+\varepsilon)$-approximate $n^s$-hops all-pairs-distances of $G$.
Each edge update requires
$\tilde{O}(
n^{\omega(1,s+\mu,1)-\mu} / \varepsilon
+ n^{\sankowski+s})$ time.
For any $I,J \subset V$, we can query the approximate distances
for the pairs $I \times J$ in
$\tilde{O}(n^{\omega(\delta_1,\mu+s,\delta_2)} / \varepsilon)$ time, 
where $\delta_1,\delta_2$ are such that $|I| = n^{\delta_1}, |J| = n^{\delta_2}$.
\end{lemma}

We now outline why \Cref{lem:overview:short hop} 
is enough to obtain \Cref{thm:overview:main}. 
The formal proof is given in \Cref{sec:apsp}.
A common technique for graph algorithms is a so-called ``hitting-sets".
More accurately, for some $h \in \mathbb{N}$, a uniformly at random chosen set 
$H \subset V$ of size $\tilde{O}(n/h)$ has w.h.p. the property 
that every shortest path with at least $h$ hops
can be decomposed into segments $s \to h_1 \to ... \to h_k \to t$,
where each $h_i \in H$ and each segment uses at most $h$ hops.
This technique goes back to Ullman and Yannakakis \cite{UllmanY91}.

Let $\hat{D}$ be an $h$-hop distance matrix, 
and denote for any $I,J \subset V$ with $\hat{D}_{I,J}$ 
the submatrix corresponding to the pairs $(u,v) \in I \times J$.
Then the length of shortest paths, using more than $h$ hops, 
can be computed via the $(\min, +)$-product
$\hat{D}_{V,H} \hat{D}_{H,H}^{|H|} \hat{D}_{H,V}$.
Using techniques from \cite{Zwick02}, 
we can compute a $(1+\varepsilon)$-approximation of
$\hat{D}_{V,H} \hat{D}_{H,H}^{|H|}$ for $h = n^s$ in
$\tilde{O}(n^{\omega(1,1-s,1-s)} / \varepsilon)$ time
during each update (this is the third term in the update time in \Cref{thm:overview:main}).

Whenever we must answer a query for the pairs $I \times J$, 
we then compute
$(\hat{D}_{V,H} \hat{D}_{H,H}^{|H|})_{I,H} \hat{D}_{H,J}$
in $\tilde{O}(n^{\omega(\delta_1,1-s,\delta_2)})$ time
and return the element-wise minimum with $\hat{D}_{I,J}$.
Thus for $\mu := 1-2s$ \Cref{lem:overview:short hop} 
implies \Cref{thm:overview:main}.

\paragraph{From Graphs to Polynomial Matrices}

The task of maintaining $h$-hop distances can be reduced to 
maintaining the {\em inverse} of a polynomial matrix 
$M \in (\F[X]/\langle X^{h+1} \rangle)^{n \times n}$ 
(i.e. a matrix whose entries are polynomials modulo $X^{h+1}$).
In general the inverse of $M$ might not exist, 
because $\F[X]/\langle X^{h+1} \rangle$ is a ring, not a field.
However, for the special case that $M$ is of the form $M = \I - A \cdot X$,
where $\I$ is the identity matrix,
the inverse is given by $M^{-1} = \sum_{i=0}^h A^i \cdot X^i$.\footnote{
This is because $X^{h+1} = 0$ in $\F[X]/\langle X^{h+1}\rangle$, so
$(\sum_{i=0}^h A^i \cdot X^i) \cdot (\I - A \cdot X) 
= \sum_{i=0}^h A^i \cdot X^i - \sum_{i=1}^{h+1} A^i \cdot X^i
= \I$.
}
A similar technique was previously used by Sankowski in \cite{Sankowski05}, 
where Sankowski used the adjoint of a polynomial matrix, instead of the inverse.
For the reduction of \Cref{lem:reduction:distanceToInverse} below, 
we need the following notation:
For a polynomial matrix $M \in (\F[X]/ \langle X^h \rangle)^{n \times n}$
we define $M^{[k]}$ to be the coefficients of $X^k$, 
so $M^{[k]} \in \F^{n \times n}$ for any $k < h$ and $M = \sum_{k=0}^{h-1} M^{[k]} X^k$.

\ifdefined\FOCSversion
\begin{lemma}[Proven in the full version]
\else
\begin{lemma}[Proven in \Cref{app:reduction:distanceToInverse}]
\fi
\label{lem:reduction:distanceToInverse}
\label{lem:distanceInverseReduction}

Let $G$ be a directed graph 
with positive integer edge weights $(c_{u,v})_{(u,v)\in E}$
and let $h$ be some positive integer.
We define $A(G) \in (\F[X]/ \langle X^h \rangle)^{n \times n}$,
such that $A_{u,v} = a_{u,v} X^{c_{u,v}}$ for each edge $(u,v)$,
$A_{v,v} = a_{v,v} X$ for every $v \in V$, and $A_{u,v} = 0$ otherwise.
Here each $a_{u,v}$ is an independent 
and uniformly at random chosen element from $\F$.
\footnote{
	Note that for $c_{u,v} \ge d$ we have $A_{u,v} = 0$,
	as if the edge would not exist.
}
	
Then, the matrix $M = \I - A(G)$ is invertible and with probability
at least $1 - hn^2/|\F|$ the following property holds:
For every $u,v \in V$ and $0 \le d < h$
the entry $(M^{-1})^{[d]}_{u,v} \neq 0$ 
if and only if 
$\dist(u,v) \le d$.

\end{lemma}

Assume we have a data-structure that allows us to maintain
$(M^{-1})^{[k]}$ for $k \in S = \{ 
	\lfloor(1+\varepsilon)^i\rfloor 
\mid 
	0 \le i \le \lceil \log_{(1+\varepsilon)}n^s \rceil
\}$.
Then the smallest $k \in S$ with $(M^{-1}_{s,t})^{[k]} \neq 0$ 
is a $(1+\varepsilon)$-approximation of $\dist(s,t)$ 
according to \Cref{lem:distanceInverseReduction}.
Thus \Cref{lem:overview:short hop} can be obtained by
creating a data-structure that maintains
$(M^{-1})^{[k]}$ for $k \in S$.

For the field $\F$ we will use $\Z_p$ for a prime of bit-length $\Theta(\log(hn))$, 
then the result is correct with high probability. 
The variable $h$ will typically be polynomial in $n$, 
so each arithmetic operation in $\Z_p$ can be performed in $\tilde{O}(1)$ time.
In summary, we can obtain an algorithm as stated in \Cref{lem:overview:short hop} by proving the following lemma:
\begin{lemma}\label{lem:overview:maintainInverse}
Let $0 \le \mu$, $0\le s$, $h := n^{s+1}$ and
$M = (\I - A \cdot X) \in (\F[X]/\langle X^{h} \rangle)^{n \times n}$
be a polynomial matrix modulo $X^{h}$.
Let $S \subset \{1,...,h\}$ be any set.

Then there exists a dynamic algorithm that supports element updates to $A$,
that requires $O(n^{s+\omega})$ field operations for the pre-processing and
$\tilde{O}(
|S| n^{\omega(1,s+\mu,1)-\mu}
+ n^{\sankowski+s})$ operations per update.
The algorithm supports queries to $(M^{-1})_{I,J}$
for any $I,J \subset [n]$, where it returns $(M^{-1})_{I,J}^{[d]}$ for all $d \in S$ in
$\tilde{O}(|S| n^{\omega(\delta_1,\mu+s,\delta_2)})$ operations.
Here $\delta_1,\delta_2$ are such that
$|I| = n^{\delta_1}, |J| = n^{\delta_2}$.
\end{lemma}

\subsection{Proof Sketch of \Cref{lem:overview:maintainInverse}}
\label{sec:overview:new}

As outlined near the end of the previous subsection,
we can obtain \Cref{lem:overview:short hop} by
proving \Cref{lem:overview:maintainInverse} via some chain of reductions.

\begin{figure*}
\centering
\begin{tabular}{llll}
Reference & Element update & Batch query & Remark\\
\hline
\cite{Sankowski05} & $\tilde{O}(n^{\omega(1,\mu,1)+s-\mu} + n^{\sankowski+s})$ & $\tilde{O}(n^{\omega(\delta_1,\mu,\delta_2)+s})$ & Maintains $(M^{-1})^{[k]}$ for $k \in [n^s]$.\\
\Cref{lem:overview:maintainInverse} & $\tilde{O}(|S| n^{\omega(1,s+\mu,1)-\mu} + n^{\sankowski+s})$ & $\tilde{O}(|S| n^{\omega(\delta_1,\mu+s,\delta_2)})$ & Maintains $(M^{-1})^{[k]}$ for $k \in S \subset [n^s]$. \\
\end{tabular}
\caption{\label{tbl:comparison}
Comparison of data-structures for maintaining $(M^{-1})^{[k]}$ for $k \in S \subset [n^s]$. 
The data-structure of \cite{Sankowski05} supports only the special case $S = [n^s]$.}
\end{figure*}

Before we will prove this result, 
we first want to compare it to other dynamic algebraic algorithms
by explaining how our algorithm manages to break 
a long-standing bottleneck of dynamic algebraic distance algorithms.
The main differences of \Cref{lem:overview:maintainInverse} 
compared to previous dynamic algebraic algorithms for distances (e.g. \cite{Sankowski05,BrandNS18})
is that our algorithm maintains $(M^{-1})^{[k]}$ for $k \in S$ for any set $S \subset [n^s]$,
whereas previous algebraic algorithms for distances
maintain $(M^{-1})^{[k]}$ for \emph{all} $k=1,2,...,n^s$, 
i.e. they were restricted to the special case $S = [n^s]$.

This difference allows us to circumvent 
a long-standing $\Omega(n^2)$ bottleneck for algebraic algorithms 
that maintain single-source distances.
Remember from the previous subsection, 
that these algebraic data-structures are used 
to maintain the $n^s$-hop distances 
of some $\tilde{O}(n^{1-s})$-sized hitting set $H$ to all other nodes $V$. 
Following the reduction of \Cref{lem:reduction:distanceToInverse},
this means maintaining the submatrix $M^{-1}_{H,V}$,
which unfortunately means that just the output-size is already $\Omega(n^2)$: 
We have a $|H| \times |V| = \Omega(n^{1-s})\times n$ sized submatrix, 
where each entry is a polynomial with $n^s$ coefficients, 
i.e. a total of $\Omega(n^2)$ field elements.
However, our algorithm does not maintain all coefficients of $M^{-1}_{H,V}$;
we just maintain the coefficients of monomials of degree $k \in S$.
So for $|S| \ll n^s$, 
the output size is smaller than $\Omega(n^s)$.
For example, when maintaining $(1+\varepsilon)$-approximate distances 
we have $|S| = O(\varepsilon^{-1} \log n)$ 
as outlined at the end of the previous subsection \ref{sec:overview:basics}.
Previous dynamic algebraic algorithms could not break this $\Omega(n^2)$ bottleneck 
as they were restricted to the special case $S = [n^s]$.

When comparing \Cref{lem:overview:maintainInverse} with \cite{Sankowski05} specifically (see \Cref{tbl:comparison}),
then our algorithm manages to replace the $n^s$ factor by having it as an argument to $\omega(\cdot, \cdot, \cdot)$,
i.e. the term $n^{\omega(1,\mu,1)+s-\mu}$ in the update time of \cite{Sankowski05} changes to $|S|n^{\omega(1,s+\mu,1)-\mu})$.
This allows us to greatly exploit fast matrix multiplication. 
Consider for example the case $\omega = 2$,
then $|S|n^{\omega(1,s+\mu,1)-\mu} = |S|n^{2-\mu}$, 
but $n^{\omega(1,\mu,1)+s-\mu} = n^{2-\mu+s}$.
So when the set $S \subset [n^s]$ has $|S| \ll n^s$, then our algorithm is a lot faster.

\paragraph{Proof idea of \Cref{lem:overview:maintainInverse}}

In order to complete our proof sketch of \Cref{thm:overview:main}, 
we are left with proving \Cref{lem:overview:maintainInverse}.
The algorithm is based on the following modified variant of the Sherman-Morrison-Woodbury identity \cite{ShermanM50,Woodbury50}:
\begin{lemma}[Paraphrased, Formal statement in \Cref{lem:incrementalShermanMorrison}]
Let $M$ be an invertible $n \times n$ matrix 
and let $M_{(t)}$ be the matrix $M$ after changing any $t$ entries, 
such that $M_{(t)}$ is invertible.

Then there exists $n$-dimensional vectors 
$u_{(1)},...,u_{(t)}$, $v_{(1)},...,v_{(t)}$ 
that are given by rescaled rows and columns of $M^{-1}$, 
such that
$$
M_{(t)}^{-1} = M^{-1} - \sum_{i=1}^t u_{(i)} v_{(i)}^\top.
$$
\end{lemma}

Assume we know $M^{-1}$ because we computed it during the pre-processing,
then one can create a dynamic algorithm by simply adding 
one new pair $u_{(t+1)},v_{(t+1)}$ for every update.
This new pair can be computed quickly, 
as they are just a row/column of the already computed $M^{-1}$.
For answering queries to $M_{(t)}^{-1}$, 
one must simply compute some entries 
of the sum $\sum_{i=1}^t u_{(i)} v_{(i)}^\top$ 
and subtract them from the corresponding entries of $M^{-1}$.
From time to time, when the sum grows too large and the queries too slow,
the data-structure will reset by computing $M_{(t)}^{-1}$ explicitly, 
i.e. we set $M \leftarrow M_{(t)}$.
(This is where the trade-off parameter $n^\mu$ for update and query time comes from.)

Performing the queries and reset operation naively 
results in the data-structure of \cite{Sankowski05} 
as presented in \Cref{tbl:comparison}.
We now outline how we are able to speed-up both of these operations.

\paragraph{Query operation}
Consider the sum $\sum_{i=1}^t u_{(i)} v_{(i)}^\top$,
then we can write it as a matrix product $UV^\top$ of $(n \times t)$-matrices $U,V$, 
where the $i$th columns are $u_{(i)}$ and $v_{(i)}$ respectively.
For answering the query of some submatrix $(M_{(t)}^{-1})_{I,J} = M^{-1}_{I,J} - (U V^\top)_{I,J}$,
we must multiply the rows $I$ of $U$ with the columns $J$ of $V^\top$.
Note however, that we only need to compute the coefficients of degree $k$ for $k \in S$,
i.e. $(U V^\top)_{I,J}^{[k]}$ instead of $(U V^\top)_{I,J}$.
This allows for some speed-up via the following lemma:

\begin{lemma}\label{lem:computeSlice}
Let $0 \le a,b,c,d$ and let
$U \in \F[X]^{n^a \times n^b}, V \in \F[X]^{n^c \times n^b}$ be polynomial matrices of degree at most $n^d$.
Then we can compute for any $k$ the $k$th coefficient $(UV^\top)^{[k]}$ in
$\tilde{O}(n^{\omega(a,b+d,c)})$ field operations.
\end{lemma}

\begin{proof}
We have
$
(UV^\top)^{[k]}
= \sum_{i=0}^d U^{[i]} V^{[k-i]\top}
= \left[ U^{[0]} \mid U^{[1]} \mid ... \mid U^{[k]} \right]
\left[ V^{[k]} \mid V^{[k-1]} \mid ... \mid V^{[0]} \right]^\top
$.
This product can be computed in $\tilde{O}(n^{\omega(a,b+d,c)})$ field operations.

\end{proof}

This way we are able to compute 
$(M_{(t)}^{-1})_{I,J}^{[k]} = (M^{-1}_{I,J})^{[k]} - (U V^\top)_{I,J}^{[k]}$ 
for all $k \in S \subset [n^s]$ in 
$\tilde{O}(|S| n^{\omega(\delta_1,s+\mu,\delta_2)})$ operations,
when $|I| = n^{\delta_1}$, $|J| = n^{\delta_2}$, $t \le n^\mu$.
This is exactly the query complexity of \Cref{lem:overview:maintainInverse}.

\paragraph{Reset operation}
After the data-structure received $n^\mu$ updates, 
the matrices $U,V$ are too large to answer the queries quickly enough.
Thus we ``reset" the algorithm by assigning $M \leftarrow M_{(t)}$
and computing $M^{-1}$ explicitly.

Note that for answering queries to 
$$(M_{(t)}^{-1})^{[k]} = (M^{-1})^{[k]} - (U V^\top)^{[k]}$$ 
for $k \in S$,
we do not need to know the entire $M^{-1}$. It is enough to know only $(M^{-1})^{[k]}$ for all $k \in S$.
Thus we can speed-up the reset by computing only those coefficients.
The number of operations required for that is just $\tilde{O}(|S|n^{\omega(1,\mu+s,1)})$,
as this is equivalent to answering a query for $I = J = [n]$.
Since this reset happens only after every $n^\mu$ updates, 
this yields the $\tilde{O}(|S|n^{\omega(1,\mu+s,1)-\mu})$ term 
in the update time of \Cref{lem:overview:maintainInverse}.

This idea of the improved reset operations
leads to the following problem: 
If we do not know all coefficients of $M^{-1}$, 
then we can not obtain the new vectors $u_{(t+1)},v_{(t+1)}$.
Previously we said that those vectors can be trivially obtained 
by just reading rows/columns of $M^{-1}$.
As we do not have all coefficients of $M^{-1}$, 
we then do not have all coefficients of $u_{(t+1)},v_{(t+1)}$
or $U,V$ either. 
Thus we can not use \Cref{lem:computeSlice} to answer queries anymore.

The solution to this problem is to run Sankowski's algorithm in parallel.
This algorithm was used in \cite{Sankowski05} to maintain exact distances,
so in graph algorithms context, this means we are running the exact distance algorithm from \cite{Sankowski05}
inside our approximate distance algorithm.
At first, this might be a bit counter-intuitive, 
as one might wonder how an approximate algorithm can be fast, 
if it needs to run an exact algorithm anyway.
We explain in the next paragraph, 
why for many problems (APSP, SSSP, diameter etc.)
our approximate algorithm is faster than Sankowski's exact one,
even though we run it internally.

\paragraph{Running exact and approximate algorithms in parallel}

By running both an exact and an approximate algorithm in parallel,
we are able to exploit their benefits to fix the other algorithm's disadvantages.

Dynamic algorithms often allow for a trade-off between update and query time,
which is something we exploit in our dynamic approximate distance algorithm,
by running the two algorithms in parallel:
\begin{itemize}
\item Sankowski's algorithm from \cite{Sankowski05} is exact, 
	with improved update time at the cost of a slower query (i.e. large choice of $\mu$ in \Cref{tbl:comparison}), 
	but it only needs to answer very few queries (one row/column per update) to obtain $u_{(t+1)},v_{(t+1)}$.
\item The other algorithm is approximate, 
	with improved query time (so we can answer large hitting set queries), 
	at the cost of a slower update.
	However, because of the approximation, this ``slower" update time is still quite fast.
\end{itemize}
By combining the two algorithms, 
they are able to compensate each other's disadvantages:
As outlined before, our approximate algorithm can not run on its own,
so we run Sankowski's exact algorithm in parallel.
However, if one were to run only Sankowski's algorithm,
then it would be very slow for maintaining the distances 
of the hitting-set to all other nodes
(i.e. $O(n^2/h)$ entries of the inverse, when maintaining $h$-hop distances).
\footnote{
	One can apply \Cref{lem:computeSlice} to turn 
	Sankowski's algorithm into an approximate algorithm 
	and thus improving the query time a bit, 
	but this approach will not result in an algorithm as fast as our approximate one, 
	because the time required to reset Sankowski's algorithm is a lot larger than ours.
}
By running it internally inside our approximate algorithm,
the exact algorithm only needs to compute a single row and column of the inverse,
so only $O(n)$ entries per update as opposed to $O(n^2/h)$.

%% file: preliminaries.tex
\section{Preliminaries}
\label{sec:preliminaries}

\paragraph{Complexity Measures}

Most of our algorithms work over any field $\F$ and their complexity
is measured in the number of arithmetic operations performed over $\F$,
i.e. the \emph{arithmetic complexity}. This does not necessarily equal
the \emph{time complexity} of the algorithm as one arithmetic operation
could require more than $O(1)$ time, e.g. very large rational numbers
could require many bits for their representation. This is why our
algebraic lemmas and theorems will always state
``in $O( \cdot )$ operations" instead of ``in $O( \cdot )$ time".
Further, $\tilde{O}(\cdot)$ hides $\polylog n$ and $\polylog \varepsilon$ factors,
so unlike the introduction, the formal proofs will no longer hide any $\log W$ factor.

For the graph applications however, when having an $n$ node graph,
we will typically use the field $\Z_p$ for some prime $p$ of order
$n^c$ for some $c$. This means each field element requires only
$O(c \log n)$ bits to be represented and all field operations can be
performed in $\tilde{O}(c)$ time in the standard model.
For our final results this $c$ will be constant, i.e. all arithmetic
operations can be performed in $\tilde{O}(1)$.

\paragraph{Notation: Identity and Submatrices}
The identity matrix is denoted by $\I$.
Let $I, J \subset [n] := \{1,...,n\}$ and $A$ be an $n \times n$ matrix,
then the term $A_{I,J}$ denotes the submatrix of $A$ consisting of the
rows with index in $I$ and columns with index in $J$.
For some $i \in [n]$ we may also just use the
index $i$ instead of the set $\{i\}$. 
For example the term $A_{[n],i}$ refers to the $i$th column of $A$.

\paragraph{Matrix Multiplication}

We denote with $O(n^\omega)$ the arithmetic complexity of multiplying
two $n \times n$ matrices.
Currently the best bound is $\omega < \matrixExponent$ \cite{Gall14a,Williams12}.

For rectangular matrices we denote the complexity of multiplying an
$n^a \times n^b$ matrix with an $n^b \times n^c$ matrix with
$O(n^{\omega(a,b,c)})$ for any $0 \le a,b,c$.
Note that $\omega( \cdot, \cdot, \cdot)$ is a symmetric function, 
so we can reorder the arguments.
Also by splitting a matrix product into several smaller products, 
we have $O(n^{\omega(a,b,c+d)}) = O(n^{\omega(a,b,c)+d})$.
The current best bounds for $\omega(1,1,c)$ can be found in \cite{GallU18}.
\ifdefined\FOCSversion
To see how to bound general $\omega(a,b,c)$, we refer to the full version \cite{BrandN19}.
\else
To see how to bound general $\omega(a,b,c)$, see \Cref{app:omega}.
\fi

\paragraph{Polynomials modulo $X^d$}

All our algebraic results use polynomials modulo $X^d$ for some positive integer $d$.
The ring of such polynomials is denoted by $\F[X] / \langle X^d \rangle$.
Given two polynomials $p, q \in \F[X] / \langle X^d \rangle$,
we can add and subtract the two polynomials in $O(d)$ operations in $\F$.
We can multiply the two polynomials in $O(d \log d)$ using fast-fourier-transformations.
If $q$ is of the form
$c - X \cdot h$, $c \in \F \setminus \{ 0 \}, h \in \F[X] / \langle X^d \rangle$,
then $q$ is invertible with
$q^{-1} = c^{-1} \sum_{k=0}^{d-1} (X h/c)^k = c^{-1} \prod_{k=0}^{\log d}  (1 + (Xh/c)^{2^k}$
so the inverse can be computed in $O(d (\ln d)^2)$ operations.
Since we typically hide polylog factors in the $\tilde{O}( \cdot )$ notation,
all arithmetic operations with polynomials from $\F[X] / \langle X^d \rangle$
can be performed in $\tilde{O}(d)$ operations in $\F$.

%

\paragraph{Polynomial Matrices}

We will work with polynomial matrices and vectors
$M \in (\F[X] / \langle X^d \rangle)^{n \times n}$,
$\vec{v} \in (\F[X] / \langle X^d \rangle)^{n}$,
so matrices and vectors whose entries are polynomials modulo $X^d$.
Products of such matrices/vectors can be performed as usual,
but since each arithmetic operations of two entries requires
$\tilde{O}(d)$ field operations, the complexity increases by
a factor of $\tilde{O}(d)$.
For example two $n \times n$ matrices can be multiplied in
$\tilde{O}(dn^\omega)$ field operations.

Note that not every matrix $M \in (\F[X] / \langle X^d \rangle)^{n \times n}$
has an inverse, (even if $\det(M) \neq 0$) 
as $\F[X] / \langle X^d \rangle$ is a ring.
However, we will only invert matrices of the form $M = \I - X \cdot A$,
where $A \in (\F[X] / \langle X^d \rangle)^{n \times n}$.
The inverse of these matrices is given via
$\sum_{i=0}^{d-1} X^k A^k = \prod_{k=0}^{\log d} (\I + A^{2^k})$,
which can be computed in $\tilde{O}(dn^\omega)$.


For a polynomial matrix $M \in (\F[X] / \langle X^d \rangle)^{n \times n}$
we define $M^{[k]}$ to be the matrix of coefficients of $X^k$.
So $M = \sum_{i=0}^{d-1} M^{[k]} X^k$ and $M^{[k]} \in \F^{n \times n}$.

\paragraph{$(1+\varepsilon)$-approximate $h$-hop distance matrix}

Given an $n$-node graph $G$ we call an $n \times n$ matrix $D$ a
$(1+\varepsilon)$-approximate $h$-hop distance matrix, if
\begin{itemize}
\item $\dist_G(u,v) \le D_{u,v} \le (1+\varepsilon) \dist_G(u,v)$, if the shortest $uv$-path uses at most $h$ hops.
\item $\dist_G(u,v) \le D_{u,v}$ for all other pairs $u,v \in V$.
\end{itemize}

%% file: small_paths.tex
\section{Algebraic Dynamic Short Hop Distances}
\label{sec:smallHop}

In this section we prove the main tool used for our new results.
This new tool allows us to maintain 
approximate bounded hop distances in a dynamic graph.
We will later extend this algorithm to work on paths of any hop length in \Cref{sec:apsp}.

The main result in this section will be the following theorem:

\begin{theorem}[$n^s$-hop distances, approximate, positive real weights]
\label{thm:dynamicSmallApproximateDistancesLargeWeightQuery}
Let $G$ be a graph with $n$ nodes and real edge weights from $[1,W]$.
Then for any $0 \le \mu, s \le 1$ and $\varepsilon > 0$
there exists a Monte Carlo dynamic algorithm that maintains
$(1+\varepsilon)$-approximate $n^s$-hops all-pairs-distances of $G$.

The preprocessing time is $\tilde{O}((n^{s+\omega} / \varepsilon) \log W)$.
Each edge update requires
$\tilde{O}(
(n^{\omega(1,s+\mu,1)-\mu} / \varepsilon^2
+ n^{\sankowski+s} / \varepsilon) \log W)$ time.

For any $I,J \subset V$, we can query the approximate distances
for the pairs $I \times J$ in
$\tilde{O}(n^{\omega(\delta_1,\mu+s,\delta_2)} / \varepsilon^2 \log W)$ time, 
where $\delta_1,\delta_2$ are such that $|I| = n^{\delta_1}, |J| = n^{\delta_2}$.
\end{theorem}

The proof will be split into two parts:
First, we will prove an equivalent result 
for graphs with integer edge weights 
in \Cref{sub:integerWeights}.
Then we will extend the result to work on
graphs with real edge weights in \Cref{sub:realWeights}.

\subsection{Exact and Approximate Distances for Integer Weights}
\label{sub:integerWeights}

We first start with the case of integer weights,
we we will later extend the algorithm to real weights. 
The integer version of \Cref{thm:dynamicSmallApproximateDistancesLargeWeightQuery}
can be formulated as follows:

\begin{theorem}\label{thm:dynamicSmallApproximateDistancesQuery}
Let $G$ be a graph with $n$ nodes and positive integer edge weights.
Then for any $0 \le s,\mu$, there exists a Monte Carlo dynamic algorithm
that maintains $(1+\varepsilon)$-approximate all-pairs-distances of $G$ upto $n^s$.

The preprocessing time is $\tilde{O}(n^{s+\omega})$.
Each edge update requires
$\tilde{O}(
s^2 n^{\omega(1,s+\mu,1)-\mu} / \varepsilon
+ s n^{\sankowski+s})$ time.

For any $I,J \subset V$, we can query the approximate distances upto $n^s$
for the pairs $I \times J$ in
$\tilde{O}(s^2 n^{\omega(\delta_1,\mu+s,\delta_2)} / \varepsilon)$ time, 
where $\delta_1,\delta_2$ are such that $|I| = n^{\delta_1}, |J| = n^{\delta_2}$.
For pairs $P \subset I \times J$ with distance larger than $n^s$, 
the returned distance is $\infty$.

\end{theorem}

The high-level idea of the algorithm for
\Cref{thm:dynamicSmallApproximateDistancesQuery}
was already given in \Cref{sec:overview}.

As previously stated, 
\Cref{thm:dynamicSmallApproximateDistancesQuery} is the result 
of maintaining the inverse of a polynomial matrix 
and using the reduction of \Cref{lem:reduction:distanceToInverse}.
As such, our first task is to create a new algorithm 
for maintaining the inverse of a polynomial matrix.

\begin{lemma}\label{lem:dynamicApproximateSlice}
Let $0 \le \mu, s$ and
$M = (\I - A \cdot X) \in (\F[X]/\langle X^{n^s} \rangle)^{n \times n}$
be a polynomial matrix modulo $X^{n^s}$.
Let $S \subset [n^s]$ be any set and
let $u(d,n)$ be a bound on update and query time of \Cref{lem:dynamicMatrixInverse}
for an $n \times n$ polynomial matrix modulo $X^{d}$.
Then there exists a dynamic algorithm that performs
$\tilde{O}(n^{s+\omega})$ operations during the pre-processing and
$\tilde{O}(
|S| n^{\omega(1,s+\mu,1)-\mu}
+ u(n^s, n))$ operations per element update to $A$.

The algorithm supports queries to $(M^{-1})_{I,J}^{[d]}$
for any $I,J \subset [n]$ and $d \in S$ in
$\tilde{O}(n^{\omega(\delta_1,\mu+s,\delta_2)})$ operations,
where $\delta_1,\delta_2$ are such that
$|I| = n^{\delta_1}, |J| = n^{\delta_2}$.
\end{lemma}

As already stated in \Cref{lem:dynamicApproximateSlice} and \Cref{sec:overview},
we build our algorithm upon the following result by Sankowski \cite{Sankowski05}:

\begin{lemma}[{\cite[Theorem 3]{Sankowski05}}]
\label{lem:dynamicMatrixInverse}
\footnote{
	\cite[Theorem 3]{Sankowski05} maintains the adjoint modulo $X^{n+1}$,
	but as stated in the proof of \cite[Theorem 6]{Sankowski05}
	the algorithm can also be used modulo $X^{n^s}$
	in which case the complexity is as stated in \Cref{lem:dynamicMatrixInverse}.
	In \cite{Sankowski05} Sankowski considered maintaining the adjoint for the case,
	where the input matrix $M$ has degree 1,
	but the algorithm can also be used to maintain the inverse for matrices of any degree bounded by $n^s$,
	as proven in \cite[Appendix C1, C2]{BrandNS18}.}

Let $0 \le \nu, s$ and
$M = (\I - A \cdot X) \in (\F[X]/\langle X^{n^s} \rangle)^{n \times n}$
be a polynomial matrix modulo $X^{n^s}$.

Then there exists a dynamic algorithm that supports element updates to $A$ in
$\tilde{O}(n^{s+\omega(1,1,\nu)-\nu} +n^{1+s+\nu})$ operations
and both row and column queries to $M^{-1}$ in $\tilde{O}(n^{1+s+\nu})$ operations.

The pre-processing requires $\tilde{O}(n^{s+\omega})$ operations.

For current $\omega$ the update and query time are $O(n^{\sankowski+s})$ for $\nu \approx 0.5285$.

\end{lemma}

The algorithm of \Cref{lem:dynamicMatrixInverse} 
could be modified to support batch queries to $M^{-1}_{I,J}$ 
with $|I| = n^{\delta_1}$, $|J| = n^{\delta_2}$ 
in $\tilde{O}(n^{s+\omega(\delta_1, \nu, \delta_2})$,
which is the result we stated in \Cref{sec:overview} \Cref{tbl:comparison},
when choosing $\nu = \min(\mu, 0.5285)$.

The high level idea of \Cref{lem:dynamicApproximateSlice} is to express the
changes to the matrix $M$ as rank-1 updates of the form $M + uv^\top$.
For such updates the new inverse of $(M+uv^\top)^{-1}$ is given via the
Sherman-Morrison identity. 

\begin{lemma}[Sherman-Morrison]\label{lem:shermanMorrison}
Let $M$ be an $n \times n$ matrix
and $u,v$ be $n$-dimensional vectors,
then 
$$
(M+uv^\top)^{-1} =
M^{-1} - M^{-1} u (1 + v^\top M^{-1}u)^{-1} v^\top M^{-1}.
$$
\end{lemma}

A dynamic algorithm receives several updates in online sequence.
For this purpose one could use the Sherman-Morrison-Woodbury identity,
however, given the online nature of the updates
(i.e. the updates are given one-by-one)
the following incremental variant of Sherman-Morrison is more useful:

\begin{lemma}\label{lem:incrementalShermanMorrison}
Let $M$ be an $n \times n$ matrix and
$u_{(1)},...,u_{(k)}$, 
$v_{(1)},...,v_{(k)}$ 
be $n$-dimensional vectors.
Define $M_{(t)} := M + \sum_{i=1}^t u_{(i)} v_{(i)}^\top$ for $t=0,...,k$,
so $M_{(t)} = M_{(t-1)} + u_{(t)} v_{(t)}^\top$.
Further define $\hat{u}_{(t)} := M_{(t-1)}^{-1}u_{(t)}$
and $\hat{v}_{(t)}^\top := (1 + v_{(t)}^\top \hat{u}_{(t)})^{-1} v_{(t)}^\top M_{(t-1)}^{-1}$
for $t=1,...,k$.

Then for all for $t=0,...,k$
$$
M_{(t)}^{-1} = M^{-1} - \sum_{i=1}^t \hat{u}_{(i)} \hat{v}_{(i)}^\top.
$$
\end{lemma}

\begin{proof}
Via \Cref{lem:shermanMorrison} we know
\ifdefined\FOCSversion
\begin{align*}
M_{(t)}^{-1}
&= (M_{(t-1)} + u_{(t)} v_{(t)}^\top)^{-1} \\
&= M_{(t-1)}^{-1} \\
&\indent\indent -  M_{(t-1)}^{-1} u_{(t)} (1 + v_{(t)}^\top M_{(t-1)}^{-1}u_{(t)})^{-1} v_{(t)}^\top M_{(t-1)}^{-1} \\
&= M_{(t-1)}^{-1} - \hat{u}_{(t)} (1 + v_{(t)}^\top \hat{u}_{(t)})^{-1} v_{(t)}^\top M_{(t-1)}^{-1} \\
&= M_{(t-1)}^{-1} - \hat{u}_{(t)} \hat{v}_{(t)}^\top
\end{align*}
\else
\begin{align*}
M_{(t)}^{-1}
&= (M_{(t-1)} + u_{(t)} v_{(t)}^\top)^{-1} \\
&= M_{(t-1)}^{-1} -  M_{(t-1)}^{-1} u_{(t)} (1 + v_{(t)}^\top M_{(t-1)}^{-1}u_{(t)})^{-1} v_{(t)}^\top M_{(t-1)}^{-1} \\
&= M_{(t-1)}^{-1} - \hat{u}_{(t)} (1 + v_{(t)}^\top \hat{u}_{(t)})^{-1} v_{(t)}^\top M_{(t-1)}^{-1} \\
&= M_{(t-1)}^{-1} - \hat{u}_{(t)} \hat{v}_{(t)}^\top
\end{align*}
\fi
so by induction $M_{(t)}^{-1} = M^{-1} - \sum_{i=1}^t \hat{u}_{(i)} \hat{v}_{(i)}^\top$, because $M_{(0)}^{-1} = M^{-1}$.
\end{proof}

We now have all tools available to prove \Cref{lem:dynamicApproximateSlice}.

\begin{proof}[Proof of \Cref{lem:dynamicApproximateSlice}]

We will first give the high-level idea:
Let $M$ be the input matrix during the pre-processing
and $M_{(k)}$ be the matrix after $k$ updates.
We will express the element updates to $M$ via rank-1 updates,
so for every update we receive a pair of vectors $u,v$,
i.e. adding $p \in \F[X]$ to entry $(i,j)$ of $M_{(k)}$
is the same as adding the outer-product
$uv^\top$ for $u = f \cdot e_i$, $v = g \cdot e_j$ 
and some $f,g \in \F[X]/\langle X^{n^s}\rangle$ with $f\cdot g = p$.

This means after $k$ updates,
we are tasked with maintaining the inverse of
$M_{(k)} = M+\sum_{i=1}^k u_{(i)} v_{(i)}^\top$,
where $u_{(i)},v_{(i)}$ is the pair of vectors that specify the $i$th update,
and each vector $u_{(i)},v_{(i)}$ has only one non-zero entry.

Thanks to \Cref{lem:incrementalShermanMorrison} we know
$$
(M+\sum_{i=1}^k u_{(i)} v_{(i)}^\top)^{-1} =
M^{-1} - \sum_{i=1}^k \hat{u}_{(i)} \hat{v}_{(i)}^\top.
$$

The high-level idea is to compute the vectors
$\hat{u}_{(i)}, \hat{v}_{(i)}$ after every update.
These vectors are useful for the following reason:
Let $\hat{U},\hat{V}$ be the $k \times n$ matrices,
where the $i$th column is $\hat{u}_{(i)}$ and $\hat{v}_{(i)}$ respectively.
Then
$\sum_{i=1}^k \hat{u}_{(i)} \hat{v}_{(i)}^\top = \hat{U}\hat{V}^\top$
and thus
$$((M+\sum_{i=1}^k u_{(i)} v_{(i)}^\top)^{-1})^{[d]} = (M^{-1})^{[d]} - (\hat{U}\hat{V}^\top)^{[d]}.$$
So by applying \Cref{lem:computeSlice} to the product $\hat{U}\hat{V}^\top$,
we can easily answer the queries.

To make sure the matrices $\hat{U},\hat{V}$ do not become too large,
we will reset the algorithm after $n^\mu$ updates.

\paragraph{Pre-processing}

We initialize \Cref{lem:dynamicMatrixInverse} for matrix $M$
and we compute $(M^{-1})^{[k]}$ for every $k \in S$.
This is done by computing $M^{-1}$ in $\tilde{O}(n^{s+\omega})$ operations.

\paragraph{Update}

Assume we handle the $k$th update, 
i.e. we receive the pair $u_{(k)},v_{(k)}$ 
and have $M_{(k)} = u{(k)}v_{(k)}^\top$.
The update routine consists of the following steps:
\begin{enumerate}
\item Compute $\hat{u}_{(k)} := M_{(k-1)}^{-1}u_{(k)}$
and $\hat{v}_{(k)} := (1 + v_{(k)}^\top \hat{u}_{(k)})^{-1} v_{(k)}^\top M_{(t-1)}^{-1}$ 
as defined in \Cref{lem:incrementalShermanMorrison}. \label{step:computeNewColumns}
\item Let $\hat{U},\hat{V}$ be the $n \times k$ matrices, where the $i$th columns is $\hat{u}_{(i)},\hat{v}_{(i)}$ respectively. \label{step:extendMatrices}
\item Update the algorithm of \Cref{lem:dynamicMatrixInverse}. \label{step:updateSankowski}
\end{enumerate}

Note that the data-structure of \Cref{lem:dynamicMatrixInverse} 
is always updated at the end in step \ref{step:updateSankowski}.
Thus at the start of the $k$th update, 
we can query rows and columns of $M_{(k-1)}^{-1}$ via the data-structure of \Cref{lem:dynamicMatrixInverse}.

This allows us to compute $\hat{u}_{(k)} := M_{(k-1)}^{-1}u_{(k)}$ 
in $\tilde{O}(n^{1+s} +u(n^s,n)) = \tilde{O}(u(n^s, n))$ operations as follows:
The vector $u_{(k)}$ has only one non-zero element, 
so $M_{(k-1)}^{-1}u_{(k)}$ is just one column of $M_{(k-1)}^{-1}$ scaled by the non-zero entry of $u_{(k)}$.
The column can be queried via \Cref{lem:dynamicMatrixInverse} in $O(u(n^s, n))$ operations, 
and multiplying each of the $n$ entries of that column 
by the non-zero entry of $u_{(k)}$ needs $\tilde{O}(n^s)$ operations.
Thus a total of $\tilde{O}(n^{1+s} +u(n^s,n))$ operations is required, which can be bounded by $\tilde{O}(u(n^s, n))$.

Likewise, $\hat{v}_{(k)} := (1 + v_{(k)}^\top \hat{u}_{(k)})^{-1} v_{(k)}^\top M_{(t-1)}^{-1}$
can be computed in $\tilde{O}(u(n^s, n))$ operations:
The vector $v_{(k)}$ has just one non-zero entry, so $v_{(k)}^\top M_{(t-1)}^{-1}$ is just one row of $M_{(t-1)}^{-1}$,
scaled by the non-zero entry of $v_{(k)}$.
This can be computed in the same way as $M_{(k-1)}^{-1}u_{(k)}$, 
except that, instead of a columns, we now query a row of $M_{(t-1)}^{-1}$ in $O(u(n^s, n))$ operations via \Cref{lem:dynamicMatrixInverse}.
The polynomial $(1 + v_{(k)}^\top \hat{u}_{(k)})^{-1}$ 
can be computed in $\tilde{O}(n^{1+s})$ as we have an inner product of two $n$ dimensional vectors of degree $n^s$,
and inverting the resulting degree $d$ polynomial needs only $\tilde{O}(n^s)$ operations.
Note that the inverse
$(1 + v_{(k)}^\top \hat{u}_{(k)})^{-1}$ exists,
because we can assume, without loss of generality,
that $v_{(k)}$ is of the form $X \cdot v$ for some $v \in (\F[X]/\langle X^{n^s} \rangle)^n$,
so $v_{(k)}$ has no constant terms.
This is because, by assumption of \Cref{lem:dynamicApproximateSlice},
element updates to $M = \I - X \cdot A$
are actually element updates to $A$,
which is multiplied with $X$.
Hence $(1 + v_{(k)}^\top \hat{u}_{(k)}) = 1 + X \cdot p(X)$,
for some polynomial $p(X)$, and it is thus invertible
(see preliminaries about inverting polynomials).

For step \ref{step:extendMatrices},
note that $\hat{u}_{(i)},\hat{v}_{(i)}$ for $i < k$ 
were already computed during the previous updates.
So for each update, the matrices $\hat{U},\hat{V}$ change 
by just adding one column, given by $\hat{u}_{(k)}$ and $\hat{v}_{(k)}$ respectively.
This means step \ref{step:extendMatrices} requires only $O(n^{1+s})$ operations.

In summary, an update requires $\tilde{O}(u(n^s,n))$ operations, because the $\tilde O(n^{1+s})$ term is subsumed.

\paragraph{Query}

When we want to compute the submatrix
$((M+\sum_{(i=1}^k u_{(i)} v_{(i)}^\top)^{-1})^{[d]}_{I,J}$
for some $d \in S$ and $I,J \subset [n]$,
then we need to compute $(\hat{U}_{I,[k]} \hat{V}^\top_{[k],J}))^{[d]}$
and subtract it from $(M^{-1})_{I,J}^{[d]}$.
The latter is known because of the pre-processing
and the former can be computed in
$\tilde{O}(n^{\omega(\delta_1,s+\mu,\delta_2)})$
operations via \Cref{lem:computeSlice},
where $\delta_1,\delta_2$ are such that
$|I| = n^{\delta_1}, |J| = n^{\delta_2}$.

\paragraph{Reset}

After upto $n^\mu$ updates, we reset the algorithm. 
For this we must compute
$((M+\sum_{i=1}^k u_{(i)} v_{(i)}^\top))^{-1})^{[d]}$
for all $d \in S$.
Afterward we set $M \leftarrow M+\sum_{i=1}^k u_{(i)} v_{(i)}^\top$.

The matrices $((M+\sum_{i=1}^k u_{(i)} v_{(i)}^\top)^{-1})^{[d]}$
can be computed the same way as in the query phase for $I=J=[n]$ in
$\tilde{O}(|S|n^{\omega(1,s+\mu,1)})$ operations.
This leads to an amortized update complexity of
$\tilde{O}(
	|S|n^{\omega(1,s+\mu,1)-\mu}
	+u(n^s, n))$
operations per update.
This can be made worst-case via the standard-technique
of running two copies of this algorithm in parallel
and spreading out the reset computation over several updates.
While one copy of the algorithm is performing the reset,
the other copy is able to answer the queries.
\ifdefined\FOCSversion
For a formal proof of this standard-technique see the full version of the paper.
\else
For a formal proof of this standard-technique see \Cref{app:worstCase}.
\fi

\end{proof}

Now that \Cref{lem:dynamicMatrixInverse} is proven, 
we can finally apply the reduction from \Cref{lem:reduction:distanceToInverse}
to obtain data-structures for distance problems.

\begin{corollary}\label{cor:maintainApproximateGraphSlices}
Let $0 \le \mu, s$ and
let $G$ be an $n$-node graph with positive integer weights.
Let $S \subset [n^s]$ be any subset.
Let $u(d,n)$ be a bound on update and query time of \Cref{lem:dynamicMatrixInverse}
for an $n \times n$ matrix modulo $X^{d}$.
Then there exists a Monte Carlo dynamic algorithm with
$\tilde{O}(s n^{s+\omega})$ pre-processing and
$\tilde{O}(
|S| s n^{\omega(1,s+\mu,1)-\mu}
+s u(n^s, n))$ worst-case update time for each edge update.

For any $I,J \subset V$ and $k \in S$,
we can query for the pairs $I \times J$ the boolean matrix,
which answers for all $s \in I$, $t \in J$,
if the distance from $s$ to $t$ is at most $k$.
Each such query requires $\tilde{O}(s n^{\omega(\delta_1,\mu+s,\delta_2)})$ time,
where $\delta_1,\delta_2$ are such that $|I| = n^{\delta_1}, |J| = n^{\delta_2}$.
\end{corollary}

\begin{proof}

Let $G$ be the given graph 
and $A(G)$ be the matrix as defined in 
\Cref{lem:reduction:distanceToInverse},
then for $M := \I - A(G)$ we have that
$(M^{-1})_{u,v}^{[k]} \neq 0$
if and only if $\dist(u,v) \ge k$.
Thus we simply maintain $M^{-1}$ via \Cref{lem:dynamicMatrixInverse}, 
where each edge update to $G$ corresponds to an element update to $A(G)$.

The extra factor $s$ in the complexity of \Cref{cor:maintainApproximateGraphSlices}
compared to \Cref{lem:dynamicMatrixInverse} comes from the fact
that we now measure the time instead of arithmetic operations.
For \Cref{lem:reduction:distanceToInverse} to hold with high probability,
we must use $\F = \Z_p$ where the prime $p$ has bit-length $\Theta(s \log n)$,
so one arithmetic operation requires $\tilde{O}(s)$ time in the standard model.

\end{proof}

\Cref{cor:maintainApproximateGraphSlices} works for some arbitrary set $S \subset [n^s]$.
To prove \Cref{thm:dynamicSmallApproximateDistancesQuery}, 
we are only left with specifying the correct set $S$ for 
\Cref{cor:maintainApproximateGraphSlices}.

\begin{proof}[Proof of \Cref{thm:dynamicSmallApproximateDistancesQuery}]
\Cref{thm:dynamicSmallApproximateDistancesQuery} is directly implied by
\Cref{cor:maintainApproximateGraphSlices} by letting
$S = \{
	\lfloor (1+\varepsilon)^k \rfloor
\mid
	0 \le k \le \lceil \log_{(1+\varepsilon)} n^s \rceil
\}$.

We simply run \Cref{cor:maintainApproximateGraphSlices} and whenever we ask for
the distances of some pairs $I \times J \subset V \times V$,
we query for every $k \in S$, if the distance is at most $k$.
This way we obtain $(1+\varepsilon)$-approximate distances, of distance upto $n^s$.

The complexity of \Cref{thm:dynamicSmallApproximateDistancesQuery}
is the same as \Cref{cor:maintainApproximateGraphSlices}.
We have $|S| = O(s/\varepsilon \log n)$, so the update time is
$\tilde{O}(
	s^2 n^{\omega(1,s+\mu,1)-\mu} / \varepsilon
	+s u(n^s, n)
)$
and the query time is
$\tilde{O}(
	s^2 n^{\omega(\delta_1,\mu+s,\delta_2)} / \varepsilon
)$.

\end{proof}

\subsection{Approximate Distances for Real Weights}
\label{sub:realWeights}

In the previous subsection we handled the case of graphs
with positive integer edge weights.
We now extend the results to the case of real edge weights.
The technique is based on the integer rounding trick used in
\cite[Lemma 8.1, Theorem 8.2]{Zwick02}.

\begin{theorem}\label{thm:integerToRealWeights}
Let $0 < \varepsilon, 0 \le \delta$.
Given a dynamic algorithm that maintains
$(1+\delta)$-approximate distances upto $3h/\varepsilon$
on graphs with integer weights,
then there exists a dynamic algorithm for
$(1+\delta)(1+\varepsilon)$-approximate $h$-hop distances
on graphs with real weights from $[1,W]$.
The update, query and pre-processing complexity all increase by a factor of $O(\log(Wn))$.
\end{theorem}

Before proving \Cref{thm:integerToRealWeights},
we observe that it immediately implies \Cref{thm:dynamicSmallApproximateDistancesLargeWeightQuery}.

\begin{proof}[Proof of \Cref{thm:dynamicSmallApproximateDistancesLargeWeightQuery}]
We can use \Cref{thm:dynamicSmallApproximateDistancesQuery}
to maintain $(1+\varepsilon)$-approximate distances
upto $3 n^s / \varepsilon$.
Via \Cref{thm:integerToRealWeights} we then obtain
$(1+\varepsilon)^2$-approximate $n^s$-hop distances for real weighted graphs.
By choosing a slightly smaller approximation factor,
we can also obtain $(1+\varepsilon)$-approximate $n^s$-hop distances.

Compared to \Cref{thm:dynamicSmallApproximateDistancesQuery} 
the complexities increase by a factor of 
$O(\varepsilon^{-1} \log \varepsilon^{-1} \log (nW))$.
Here the factor $O(\varepsilon^{-1} \log \varepsilon^{-1})$ 
comes from maintaining the distance upto $O(n^s / \varepsilon)$ 
and the factor $O(\log (nW))$ comes from \Cref{thm:integerToRealWeights}.
\end{proof}

\begin{lemma}\label{lem:scalingTechnique}
Let $G = (V,E)$ be a graph with $n$ nodes and real edge weights from $[1,W]$.
For any $0 < A,B$ define $G' = (V, E')$ to be the graph with
$
	E' = \{
		(u,v) \in E
	\mid
		c_{u,v} \le B
	\}
$
and integer edge weights $c'_{u,v} = \lceil A c_{u,v} / B \rceil$.

Then for any path from $s$ to $t$ in $G$ of length
$\dist_G(s,t) \le B$ let $h$ be the number of hops.
We have
$\dist_G(s,t) \le (B/A) \dist_{G'}(s,t) \le \dist_G(s,t) + (B/A) h$.
\end{lemma}

\begin{proof}
Since $\dist_G(s,t) \le B$, all edges used by the path in $G$ also exist in $G'$.
Because of the rounding, we have
$\dist_G(s,t) \le (B/A) \dist_{G'}(s,t)$
and each used edge can cause an error of at most $B/A$,
so $(B/A) \dist_{G'}(s,t) \le \dist_G(s,t) + (B/A) h$.
\end{proof}

\begin{lemma}\label{lem:parallelScaledGraphs}
Let $0 \le s$, $\varepsilon > 0$ and
let $G = (V,E)$ be a graph with $n$ nodes and real edge-weights from $[1,W]$.
Define $\lceil \log_2 nW \rceil$ graphs $G_i$ as in \Cref{lem:scalingTechnique} for
$B_i = 2^i$,
$A = 2 n^s 1/\varepsilon$
for $i = 1,...,\lceil \log_2 nW \rceil$.

Then for any pair $s,t \in V$ we have 
$\dist_G(s,t) \le \min_{i} (B_i / A) \dist_{G_i}(s,t)$ 
and if the shortest $st$-path uses at most $n^s$ hops, 
then we also have
$$ \min_{i} (B_i / A) \dist_{G_i}(s,t) \le (1+\varepsilon) \dist_G(s,t).$$
\end{lemma}

\begin{proof}
The first inequality
$$\dist_G(s,t) \le \min_{i} (B_i / A) \dist_{G_i}(s,t)$$
follows directly from \Cref{lem:scalingTechnique}.
The second inequality follows from the following observation:
Let $i$ be such that $2^{i-1} \le \dist_G(s,t) \le 2^i$ and
let $h$ be the number of hops for the shortest $st$-path.
Then this path in $G$ also exists in $G_i$ and
$
	(B_i /A) \dist_{G_i}(s,t)
	\le \dist_G(s,t) + (B_i/A) h
	\le \dist_G(s,t) + \varepsilon \dist_G(s,t) n^{-s} h$.
So if the number of hops $h$ is at most $n^s$, then we obtain the promised $(1+\varepsilon)$-approximation.

\end{proof}

\begin{proof}[Proof of \Cref{thm:integerToRealWeights}]
Consider the graphs $G_i$ for $i = 1,...,\lceil \log nW \rceil$ from \Cref{lem:parallelScaledGraphs}.
The largest $n^s$-hop distance in any $G_i$ is bounded by $n^sA = 3n^s/\varepsilon$.
So when the dynamic algorithm for integer weights 
can maintain the distance upto $3n^s/\varepsilon$, 
then it can maintain $n^s$-hop distances for all graphs $G_i$ for $i = 1,...,\lceil \log nW \rceil$.

Thus we let the given algorithm run on all $G_i$
and for any distance query, 
we return $\min_i (B_i/A) \dist_{G_i}(s,t)$ as in \Cref{lem:parallelScaledGraphs}.
This yields $(1+\varepsilon)(1+\delta)$-approximate $n^s$-hop distances,
because the distances in each $G_i$ are only maintained $(1+\delta)$-approximately.
\end{proof}

\paragraph{Non-oblivious adversaries}

Note that all the graph algorithm of \Cref{sec:smallHop} work against non-oblivious adversaries, 
i.e. updates are allowed to depend on the query results.
The only random choices are the random field elements for the non-zero entries of the matrix
in \Cref{cor:maintainApproximateGraphSlices}.
The correct return values of \Cref{cor:maintainApproximateGraphSlices} are
uniquely determined by the input graph and the query-input $(I,J,k)$.
The actual values, returned by \Cref{cor:maintainApproximateGraphSlices} ,
are correct with high probability, so with high probability the returned values for any query 
do not leak any information about the random choices.

%% file: apsp.tex
\section{Results for All-Pairs-Distances}
\label{sec:apsp}

In this section we will prove the result of \Cref{thm:intro-main},
as presented in the introduction.
The result is split into two theorems, 
one for directed, weighted graphs 
and one for undirected graphs with small integer weights.

For directed, weighted graphs we will prove the following:

\begin{theorem}[Approximate APSP, Real weights, Directed, Queries]
\label{thm:weightedAPSPQuery}
Let $G$ be a directed graph 
with $n$ nodes 
and real weights from $[1,W]$.
Then for any $0 \le s \le 1$ 
and $\varepsilon > 0$
there exists a Monte Carlo dynamic algorithm that maintains
$(1+\varepsilon)$-approximate all-pairs-distances of $G$ in
$\tilde O(	
	(n^{\sankowski+s} / \varepsilon +
		n^{\omega(1,1,1-s)+1-2s}/\varepsilon^2
		+n^{\omega(1-s,1-s,1)}/\varepsilon^2
	) \log W
)$ update time.
We can query for any $I,J \subset V$ the distances of the pairs $I \times J$ in
$\tilde O(n^{\omega(\delta_1,1-s,\delta_2)} / \varepsilon^2 \log W)$
time, where $\delta_1,\delta_2$ are such that $|I| = n^{\delta_1},|J| = n^{\delta_2}$.
The pre-processing requires $\tilde{O}(n^{s+\omega} / \varepsilon \log W)$ time.

For current $\omega$ and $s \approx \sWeiDirAPSPquery,\mu \approx \muWeiDirAPSPquery$
the update time is $\tilde O(n^{\updateWeiDirAPSPquery} / \varepsilon^2 \log W)$,
with query time for a single pair is $\tilde O(n^{\queryWeiDirAPSPquery}/\varepsilon^2  \log W)$.
The pre-processing time is $\tilde{O}(n^{\preWeiDirAPSPquery}  \log W)$.
\end{theorem}

For undirected graphs we will show the following result:

\begin{theorem}[Approximate APSP, Unweighted (or with $W$), Undirected, Queries]
\label{thm:undirectedApproximateDistancesQuery}
Let $G$ be a directed graph with $n$ nodes 
and integer weights from $\{1,...,W\}$ where $W = n^\ell$.
Then for any $0 \le s \le 1$ and $\varepsilon > 0$
there exists a Monte Carlo dynamic algorithm that maintains
$(1+\varepsilon)$-approximate all-pairs-distances of $G$ in
$\tilde O(	
n^{\omega(1,1,s+\mu + \ell)+\mu} / \varepsilon
	+ u(n^{s+\ell}, n)
	+ n^{\omega(1-s,s+\mu+\ell,1)} / \varepsilon^2
	+ n^{(1-s)\omega} / \varepsilon^{1+\omega}))$ update time.
We can query for any $I,J \subset V$ the distances of the pairs $I \times J$ in
$\tilde{O}(n^{\omega(\delta_1,s+\mu+\ell,\delta_2)}/ \varepsilon)$
time, where $\delta_1,\delta_2$ are such that $|I| = n^{\delta_1},|J| = n^{\delta_2}$.
The pre-processing requires $\tilde{O}(Wn^{s+\omega} / \varepsilon)$ time.

For current $\omega$ the pre-processing and update time for an unweighted graph are
$\tilde{O}(n^{\preUnwUndAPSPquery} / \varepsilon)$ and
$\tilde{O}(n^{\updateUnwUndAPSPquery} / \varepsilon^{3.373})$ respectively
with $s \approx \sUnwUndAPSPquery$ and $\mu \approx \muUnwUndAPSPquery$.
The query time for a single pair is $O(n^{\queryUnwUndAPSPquery} / \varepsilon)$.
\end{theorem}

The high-level idea of all our algorithms is to maintain distances for short hop paths using
\Cref{thm:dynamicSmallApproximateDistancesLargeWeightQuery}
(and \Cref{thm:dynamicSmallApproximateDistancesQuery} for integer edge weights),
and use hitting set arguments to compute distances for large hop paths.

Hitting set arguments, as introduced in \cite{UllmanY91}, 
allow us to decompose paths with many hops into segments with fewer hops, 
specifically the following lemma will be an important tool:
\begin{lemma}\label{lem:allPathsHittingSet}
Let $G$ be a graph with $n$ nodes and let $H \subset V$ be a random subset of size $c \frac{n}{d} \ln n$.

With probability at least $1-n^{2-c}$ we have that:
For every $s,t \in V$, where the shortest $st$-path uses at least $d$ hops,
this shortest $st$-path can be decomposed into segments
$s \to h_1 \to h_2 \to ... \to h_k \to t$,
where $h_i \in H$ for every $i=1,...,k$
and each segment uses at most $d$ hops.

\end{lemma}

This lemma implies, that to compute some shortest $st$-path with many hops,
we only need to know the distances of paths with few hops between 
the pairs $(\{s\} \cup H) \times (H \cup \{ t\})$.
The idea of these hitting-set arguments goes back to \cite{UllmanY91}.
\ifdefined\FOCSversion
The proof for \Cref{lem:allPathsHittingSet} can be found in the full version.
\else
The proof for \Cref{lem:allPathsHittingSet} can be found in \Cref{app:hittingSet}.
\fi

The other important ingredient for this section are the following two results by \cite{Zwick02},
which allows us to compute approximate $(\min, +)$-products and approximate distances.

\begin{lemma}[{\cite[Theorem 8.2]{Zwick02}}]
\label{lem:allPairsApproximateDistances}
Let $G$ be a directed graph with $n$ nodes and real edge weights in $[1,W]$,
then we can compute all-pairs-distances of $G$ in $\tilde{O}(n^\omega / \varepsilon \log W)$ time.

Or in other words, we can compute for any $n \times n$ matrix $D$ with entries in $[1,W]$,
a $(1+\varepsilon)$-approximation of the $n$th power of $D$ using the $(\min,+)$-product.
\end{lemma}

The original result \cite[Theorem 8.1]{Zwick02} for the approximate $(\min,+)$-products
is stated for square matrices, but it can also be used for rectangular matrices:

\begin{lemma}[{\cite[Theorem 8.1]{Zwick02}}]
\label{lem:approximateMinPlusProduct}
Let $A$ be an ${n^a \times n^b}$ and
$B$ be an ${n^b \times n^c}$ matrix,
each with real entries from $[1, W]$,
then we can compute a $(1+\varepsilon)$-approximate $(\min, +)$-product $A \star B$
in $\tilde{O}(n^{\omega(a,b,c)} /\varepsilon \log W)$ time.
\end{lemma}

\subsection{Weighted Approximate Distances (Proof of \Cref{thm:weightedAPSPQuery})}

We will start with the results for weighted all-pairs-distances.
The following lemma allows us to extend the short hop distances of
\Cref{thm:dynamicSmallApproximateDistancesLargeWeightQuery}
to large hop distances.

\begin{lemma}\label{lem:smallHopToLargeHop}
Let $G$ be a directed graph with $n$ nodes and real weights from $[1,W]$.
Assume we are already given an
$(1+\delta)$-approximate $h$-hop distance matrix $D$.
Let $H \subset V$ be a random subset of size $c (n/h) \ln n$.
Define 
\begin{align}
\hat{D} := D_{V,H} \star D_{H,H}^{(\star |H|)} \star D_{H,V}, \label{eq:minPlusProduct}
\end{align}
where $\star$ is a $(\min,+)$-product and $(\star |H|)$ is the $|H|$th power using $(\min,+)$-products.

Then for any $u,v \in V$ we have $\dist_G(u,v) \le \hat{D}_{u,v}$.
Further, with probability at least $1-n^{2-c}$, we have
for any $u,v \in V$, where the shortest $uv$-path has at least $h$ hops,
that $\hat{D}_{u,v} \le (1+\delta) \dist_G(u,v)$.
Consequently, we can obtain approximate all-pairs-distances for any number of hops,
by taking the entry-wise minimum of $D$ and $\hat{D}$.

\end{lemma}

\begin{proof}
According to \Cref{lem:allPathsHittingSet} we have that (w.h.p)
every shortest path using at least $h$ hops
can be decomposed into segments
$s \to h_1 \to h_2 \to ... \to t$
where each $h_i \in H$ and each segment uses at most $n^s$ hops.
Thus we obtain the distances of all pairs $(u,v)$
where the shortest $uv$-path uses at least $h$ hops,
by computing the $(\min, +)$-product
\begin{align*}
D_{V,H} \star D_{H,H}^{(\star |V|)} \star D_{H, V}
\end{align*}
where $(\star |V|)$ refers to the $|V|$th power using the $(\min, +)$-product.
Note that the power $D_{H,H}^{(\star |V|)}$ can be reduced to $D_{H,H}^{(\star |H|)}$,
because the matrix $D_{H,H}$ can be considered
an edge weight matrix of a matrix on node set $H$
and then $D_{H,H}^{(\star |V|)}$ is the all-pairs-distance matrix.
Since every shortest path in that graph can use at most $|H|$ hops,
the $|H|$th power is enough.

\end{proof}

Note, for constant $c$ and $h = n^s$
(so $|H| = \tilde{O}(n^{1-s})$)
we can compute $D_{H,H}^{(\star |H|)}$ approximately
via \Cref{lem:allPairsApproximateDistances}
in $\tilde{O}(n^{(1-s)\omega} / \varepsilon \log W)$ time.
The remaining products of \eqref{eq:minPlusProduct}
can be computed approximately in
$\tilde{O}(n^{\omega(1,1-s,1)}/\varepsilon \log W)$ time
via \Cref{lem:approximateMinPlusProduct}.

\begin{proof}[Proof of \Cref{thm:weightedAPSPQuery}]
We maintain the $(1+\varepsilon)$-approximate $n^s$-hop distances via
\Cref{thm:dynamicSmallApproximateDistancesLargeWeightQuery}.
We want to compute \eqref{eq:minPlusProduct} of \Cref{lem:smallHopToLargeHop},
but split the computation of the $(\min,+)$-product into two parts.
After every update we sample a random set of nodes
$H \subset V$ of size $\tilde{O}(n^{1-s})$
as in \Cref{lem:smallHopToLargeHop},
and query the distances for the pairs $H \times V$ and $V \times H$.
Let $D_{H,V}, D_{V,H}$ be the obtained distance matrices,
then we compute a $(1+\varepsilon)$-approximation of the $(\min, +)$-product
$D_{V,H} \star D_{H,H}^{(\star |H|)}$ via \Cref{lem:approximateMinPlusProduct}.
The update time is thus
\begin{align*}
\tilde{O}(
	&\underbrace{
		  n^{\omega(1,s+\mu,1)-\mu} / \varepsilon^2
		+ u(n^s / \varepsilon, n))
		\log W
	}_{\text{\Cref{thm:dynamicSmallApproximateDistancesLargeWeightQuery}}} \\
	+& \underbrace{
		n^{\omega(1,\mu+s,1-s)} / \varepsilon^2
	\log W}_{\text{query }D_{V,H},D_{H,V}} 
	+ \underbrace{
		n^{\omega(1-s,1-s,1)} / \varepsilon \log W
	}_{\text{compute } D_{H,H}^{(\star |H|)} \star D_{H,V}}
).
\end{align*}

\paragraph{Queries}

When there is some query for the distances of the pairs $I \times J$,
then we compute $\hat{D} := D_{I,H} \star D_{H,H}^{(\star |H|)} \star D_{H,J}
= (D_{V,H} \star D_{H,H}^{(\star |H|)})_{I,H} \star D_{H,J}$,
where $D_{V,H} \star D_{H,H}^{(\star |H|)}$
was already computed during the last update.
Each entry yields the distance of the pair $(u,v) \in I \times J$,
if the shortest $uv$-path uses at least $n^s$ nodes
(see \Cref{lem:smallHopToLargeHop}).
Thus we also query the $n^s$-hop distances for the pairs $I \times J$
via \Cref{thm:dynamicSmallApproximateDistancesLargeWeightQuery}
and return for each pair $(u,v) \in I \times J$
the minimum of the two distances $\hat{D}_{u,v}$ and $D_{u,v}$.
The query time is thus
$$\tilde{O}(
	\underbrace{
		n^{\omega(\delta_1, s + \mu, \delta2)}/\varepsilon^2 \log W
	}_{
		\text{query } D_{I,J}
	} + \underbrace{
		n^{\omega(\delta_1,1-s,\delta_2)} / \varepsilon \log W
	}_{
		\text{compute } \hat{D}
	}),$$
where $\delta_1,\delta_2$ are such that $n^{\delta_1} = |I|$ and $n^{\delta_2} = |J|$.

Technically this maintains a $(1+\varepsilon)^2$ approximation,
but we can simply choose a slightly smaller approximation for
\Cref{thm:dynamicSmallApproximateDistancesLargeWeightQuery}
and \Cref{lem:approximateMinPlusProduct}
to obtain a $(1+\varepsilon)$ approximation.
Further, for $1-s = \mu + s$ we obtain the update/query complexties as stated in \Cref{thm:weightedAPSPQuery}.

Note that \Cref{thm:dynamicSmallApproximateDistancesLargeWeightQuery}
works against non-oblivious adversaries,
so \Cref{thm:weightedAPSPQuery} works against
non-oblivious adversaries as well,
because we sample a new random hitting-set $H$ after every update.

\end{proof}

We can maintain all-pairs-distances explicitly, 
by querying all distance after every single update.
Thus we obtain the following result:

\begin{theorem}[Approximate APSP, Real weights, Directed, Almost-$n^2$]
\label{thm:weightedAPSP}
Let $G$ be a directed graph with $n$ nodes and real weights from $[1,W]$.
Then for any $\varepsilon > 0$ there exists a Monte Carlo dynamic algorithm
that maintains $(1+\varepsilon)$-approximate all-pairs-distances of $G$ in
$\tilde O(n^{\omega(1,1,0.5)}/\varepsilon^2 \log W)$ update time.
The pre-processing requires $\tilde{O}(n^{1/2 + \omega} / \varepsilon \log W)$ time.

For current $\omega$ the update time is $\tilde O(n^{2.045} / \varepsilon^2 \log W)$\cite{GallU18}.
\end{theorem}

\begin{proof}
\Cref{thm:weightedAPSP} is implied by \Cref{thm:weightedAPSPQuery},
by performing a query to $I = J = V$ after every update.
In that case the update and query time are balanced
for $\mu = 0, s = 0.5$, which yields
$\tilde{O}(n^{\omega(1,0.5,1)} / \varepsilon^2 \log W) = \tilde{O}(n^{2.044183} / \varepsilon^2 \log W)$
update time.
\end{proof}

To obtain an algorithm for dynamic single-source distances,
we could just use \Cref{thm:weightedAPSPQuery} 
to query the distances from the source-node.
However, the update time can be slightly improved as follows:

\begin{theorem}[Approximate SSSP, Real weights, Directed]
\label{thm:SSSP}
Let $G$ be a directed graph with $n$ nodes and real weights from $[1,W]$.
Then for any $0 \le s, \mu \le 1$ and $\varepsilon > 0$
there exists a Monte Carlo dynamic algorithm that maintains
$(1+\varepsilon)$-approximate all-pairs-distances of $G$ in
$\tilde O(
	(
		 n^{\sankowski+s}
		+n^{\omega(1,1,s+\mu)-\mu}
		+n^{\omega(1-s,\mu+s,1)}
	)/\varepsilon^2 \log W
)$ update time.
The pre-processing requires $\tilde{O}(n^{s+\omega} / \varepsilon \log W)$ time.

For current $\omega$ and $s \approx \sSSSP, \mu \approx \muSSSP$
the update time is $\tilde O(n^{\updateSSSP} / \varepsilon^2 \log W)$.
\end{theorem}

\begin{proof}[Proof of \Cref{thm:SSSP}]
Let $v \in V$ be the source node
for which we want to maintain single source distances.
We maintain distances upto $n^s$ hops via
\Cref{thm:dynamicSmallApproximateDistancesLargeWeightQuery}.

After every update we sample a hitting-set $H$
as in \Cref{lem:smallHopToLargeHop}.
We query the approximate $h$-hop distances $D_{H\cup\{v\},V}$
via \Cref{thm:dynamicSmallApproximateDistancesLargeWeightQuery}
and construct a graph $G_H$ with node set $V$
and edges between every $(u,w) \in (H \cup \{v\}) \times V$
with cost $D_{u,w}$.

Via \Cref{lem:allPathsHittingSet} we know
that w.h.p. the single-source distances rooted at $v$ in $G_H$
are the same as in $G$,
so we can compute them in
$O(|V||H|) = \tilde{O}(n^{2-s})$ time using Dijkstra.
The time for constructing $G_H$ and running Dijkstra
is subsumed by querying $D_{H \cup \{v\}, V}$,
so the update time for the single-source algorithm is
$
\tilde{O}(
	(n^{\omega(1,s+\mu,1)-\mu} / \varepsilon^2
	+ u(n^s / \varepsilon, n)
	+ n^{\omega(1-s,s+\mu,1)} / \varepsilon^2) \log W
).
$

As \Cref{thm:dynamicSmallApproximateDistancesLargeWeightQuery}
works against non-oblivious adversaries,
and we sample a new hitting-set $H$ after every update,
\Cref{thm:SSSP} works against non-oblivious adversaries as well.
\end{proof}

\subsection{Unweighted Undirected Approximate Distances (Proof of \Cref{thm:undirectedApproximateDistancesQuery})}

For undirected graphs we exploit the idea from \cite{RodittyZ12}. The high-level idea is as follows:
We have a hitting set $H$ and want to compute the $st$-distance, while assuming the shortest path uses at least $n^s$ hops.
The path can be split into segments $s \to h_1 \to ... \to h_k \to t$
but to find this path, we would usually need to figure out which element of $H$ is $h_1$.
In the case of unweighted, undirected graphs,
this can be solved approximately by choosing any $x,y \in H$ with distance
$\dist(s,x), \dist(y,t) \le n^s/2 \varepsilon$.
Then we know $\dist(x,y) \le \dist(s,t) + n^s \varepsilon$,
because we can find a $xy$-path via the segments $x \to s \to t \to y$,
and also have $\dist(s,t) \le \dist(x,y) + n^s \varepsilon$,
because we can fine a $st$-path via the segments $s \to x \to y \to t$.
Thus we have
$(1-\varepsilon) \dist(x,y) \le \dist(s,t) \le (1+\varepsilon) \dist(x,y)$
for unweighted undirected graphs where the shortest $st$-path uses at least $n^s$ hops.

\begin{lemma}\label{thm:undirectedLongPaths}
Let $G$ be an undirected graph with $n$ nodes and integer edge weights in $\{1,2,...,W\}$.
Let $H \subset V$ be a random subset of size
$2cn/(h \varepsilon) \ln h$ for some constant $c > 0$.
Assume we are given $(1+\varepsilon)$-approximate distances $D_{H \times V}$, 
where all pairs $u,v$ with $\dist_G(u,v) > Wh$ 
are allowed to have $D_{u,v} > (1+\varepsilon) Wh$.

Then we can construct a distance oracle $\hat{D}$ in
$\tilde{O}((n/h)^{\omega} / \varepsilon^{1+\omega})$ time,
such that each query $\hat{D}(u,v)$ for any $u,v \in V$ requires only $O(1)$ time,
and with probability at least $n^{2-c}$:
\begin{itemize}
\item $\dist_G(u,v) \le \hat{D}_{u,v}$
\item $\hat{D}_{u,v} \le (1+\varepsilon)^3 \dist_G(u,v)$ if $\dist_G(u,v) \ge Wh$.
\end{itemize}

\end{lemma}
 
\begin{proof}
According to \Cref{lem:allPathsHittingSet},
with probability at least $n^{2-c}$,
every shortest path using at least $0.25 h\varepsilon$ hops
can be decomposed into segments
$s \to h_1 \to h_2 \to ... \to t$ where each $h_i \in H$
and each segment uses at most $0.25 h\varepsilon$ hops

Further, any $h$-hop path can have cost at most $Wh$,
thus the given $(1+\varepsilon)$-approximate distances upto $Wh$
are also $(1+\varepsilon)$-approximate $h$-hop distances.
This means can compute $(1+\varepsilon)^2$-approximate distances
for the pairs $H \times H$ at extra cost
$\tilde{O}((n/h)^\omega / \varepsilon^{\omega+1})$
by computing $\Delta := D_{H,H}^{(\star |H|)}$.

Next, we assign each node $v \in V$ some node $x_v \in H$
with $D_{v,x_v} \le 0.25 Wh\varepsilon$.
This requires only $O(|H|n) = O(n^2/(h\varepsilon))$ time for all $v \in V$ together.
(We will later explain what happens to nodes, where no such $x_v \in H$ exists.)

Note that this also implies $\dist(u,x_u),\dist(v,x_v) \le 0.25 Wh\varepsilon$
and since the graph is undirected, this leads to
\begin{align*}
\dist(x_u,x_v)
&\le \dist(x_u,u) + \dist(u,v) + \dist(v,x_v) \\
&\le \dist(u,v) + 0.5 Wh\varepsilon,\\
\dist(u,v)
&\le \dist(u,x_u) + \dist(x_u,x_v) + \dist(x_v, v) \\
&\le \dist(x_u,x_v) + 0.5 Wh\varepsilon.
\end{align*}
Thus for any pair $u,v$ with $\dist(u,v) \ge Wh$ we have:
$$\dist(u,v) \le \Delta_{x_u,x_v} + 0.5Wh\varepsilon \le \dist(u,v) (1+\varepsilon)^3.$$

If some node $u \in V$ has no $x_u \in H$,
then there also is no $v \in V$ with $Wh \le \dist(u,v) < \infty$,
as otherwise there should be some $h \in H$ with
$\dist(u,h) \le Wh\varepsilon$ along the path from $u$ to $v$.
Thus for some distance query for a pair $u,v$,
we either return $\Delta_{x_u,x_v} + 0.5Wh\varepsilon$
or $\infty$ if there is no $x_u$ or $x_v$.
\end{proof}

\begin{proof}[Proof of \Cref{thm:undirectedApproximateDistancesQuery}]
Let $0 \le s \le 1$ be some parameter.
We maintain the $(1+\varepsilon)$-approximate distances $D$ upto $Wn^s$ via
\Cref{thm:dynamicSmallApproximateDistancesQuery}.
After every update we also construct the distance oracle $\hat{D}$
from \Cref{thm:undirectedLongPaths} for a new random hitting-set $H$
of size $\tilde{O}(n^{1-s} / \varepsilon)$.
To construct this oracle, we need to query the distances $D_{H,V}$ for the pairs $V \times H$,
so for $W = n^\ell$ the update time becomes $\tilde{O}(
	n^{\omega(1,1,s+\mu + \ell)+\mu} / \varepsilon
	+ u(n^{s+\ell}, n)
	+ n^{\omega(1-s,s+\mu+\ell,1)} / \varepsilon^2
	+ n^{(1-s)\omega} / \varepsilon^{1+\omega})$,

When answering some query for all pairs in some set $I \times J$,
we compute $D_{I,J}$ and take the entry-wise minimum with $\hat{D}_{I,J}$.
This way we obtain $(1+\varepsilon)^3$ approximate distances in
$\tilde{O}(n^{\omega(\delta_1,s+\mu+\ell,\delta_2)})$ time.
By choosing a slightly smaller approximation factor in
\Cref{thm:dynamicSmallApproximateDistancesQuery,thm:undirectedLongPaths}
this can be made $(1+\varepsilon)$-approximate.

Note that \Cref{thm:dynamicSmallApproximateDistancesQuery}
works against non-oblivious adversaries,
so \Cref{thm:undirectedApproximateDistancesQuery} works against
non-oblivious adversaries as well,
because we sample a new random hitting-set $H$ after every update.
\end{proof}

\begin{theorem}[Approx APSP, Unweighted (or with $W$), Undirected, $n^2$]
\label{thm:undirectedApproximateDistances}
Let $G$ be an undirected graph with $n$ nodes and integer weights from $\{1,...,W\}$.
Then for any $0 < \varepsilon$ 
there exists a Monte Carlo dynamic algorithm that maintains 
$(1+\varepsilon)$-approximate all-pairs-distances of $G$ in
$
\tilde{O}( Wn^{\weightUnwUndAPSP} / \varepsilon + n^2 / \varepsilon^{1+\varepsilon})
$ update time.
The pre-processing requires $\tilde{O}(Wn^{\preUnwUndAPSP} / \varepsilon)$ time.

\end{theorem}

\begin{proof}
\Cref{thm:undirectedApproximateDistances} follows from
\Cref{thm:undirectedApproximateDistancesQuery}
by querying all distances for $I = J = V$ after every update.
For $\mu = 0$ and maximum edge weight $W = n^\ell$ the update time is
$$\tilde{O}(
	n^{\omega(1,1,s+\ell)} / \varepsilon
	+ n^{(1-s)\omega} / \varepsilon^{1+\omega}).$$

By setting $s = 1-2/\omega \approx 0.157$ we have $n^{(1-s)\omega} = n^2$.
For small $W = O(n^{0.156})$, the update time can be bounded by
$\tilde{O}(n^2 / \varepsilon^{1+\omega})$, because $\omega(1,1,s+\ell) = 2$.
For $W = \Omega(n^{0.156})$ the $n^{\omega(1,1,0.157+\ell)} / \varepsilon$ term can be bounded by $
\tilde{O}(
	n^{\omega(1,1,0.157+\ell)} / \varepsilon
)
=
\tilde{O}(
	n^{2 + \ell - 0.156} / \varepsilon
)
=
\tilde{O}(
	Wn^{\weightUnwUndAPSP} / \varepsilon)$.
Thus for any $W$, we have update time $\tilde{O}( Wn^{\weightUnwUndAPSP} / \varepsilon + n^2 / \varepsilon^{1+\varepsilon})$.

The pre-processing requires $O(Wn^{\omega+s}) = O(Wn^{\preUnwUndAPSP})$ time.
\end{proof}

%% file: diameter.tex
\section{Results for Diameter, Radius and Eccentricities}
\label{sec:diameter}

In this section we will prove several results for dynamic diameter.
The main result is the following nearly $(1.5+\varepsilon)$-approximate algorithm for diameter and radius.

\begin{theorem}[Nearly $(1.5+\epsilon)$-Approx Diameter/Radius, Unweighted (only), Directed, Sub-$n^2$]
\label{thm:1.5approxDiam}
Let $G$ be an unweighted directed graph with $n$ nodes.
Then for any $0 \le s,\mu \le 1$ and $\varepsilon > 0$
there exists a Monte Carlo dynamic algorithm that maintains a nearly $(1.5+\varepsilon)$-approximation $\tilde{D}$ of the diameter of $G$, such that
$$
\left(\frac{2}{3}-\varepsilon\right) \diam(G) - 1/3
\le \tilde{D}
\le (1+\varepsilon) \diam(G).
$$
The update time of the algorithm can be bounded by
\begin{align*}
\tilde{O}(&
(
	n^{\omega(1,1,s+\mu)-\mu}
	+ n^{\omega(0.5,s+\mu,1)}
)/\varepsilon\\
&+ n^{\sankowski+s}
+ n^{\omega(1-s,\mu+s,1-s)}/\varepsilon^2
+ n^{(1-s)\omega}/\varepsilon^{1+\omega}
).
\end{align*}
The pre-processing requires $\tilde{O}(n^{s+\omega} / \varepsilon)$ time.
If the graph is undirected, then we can also maintain a nearly $(1.5+\varepsilon)$-approximate radius $\tilde{R}$, such that
$$\radius(G) / (1+\varepsilon)
\le \tilde{R}
\le (1.5+\varepsilon) \radius(G) + 2/3,
$$
For current $\omega$ and $\mu \approx \muRadius s \approx \sRadius$ 
the update time is $O(n^{\updateRadius}/\varepsilon^{1+\omega})$.
\end{theorem}

We will prove this result in \Cref{sub:approximateDiameter}. 
There we will also present results for dynamic eccentricities.
Before proving \Cref{thm:1.5approxDiam} in \Cref{sub:approximateDiameter},
we will prove some results for large diameter graphs in \Cref{sub:largeDiameter},
which can also be used to maintain a $(1+\varepsilon)$-approximation.

Many of the results we prove and use in this section hold only for strongly connected graphs,
which is why we require the following lemma:

\begin{lemma}[{\cite[Corollary C.17]{BrandNS18},\cite[Theorem 3]{Sankowski04}}]
\label{lem:strongConnectivity}
For every $\mu > 0$ there exists a Monte Carlo dynamic algorithm
that can detect if a graph is strongly connected.
The update time per edge update is
$O(n^{\omega(1,1,\mu)-\mu} + n^{1+\mu}) = O(n^{\sankowski})$
and the pre-processing requires $O(n^\omega)$ time.\footnote{
	Note that the update complexity is subsumed by \Cref{lem:dynamicMatrixInverse}.
	All algorithms in \Cref{sec:smallHop} use \Cref{lem:dynamicMatrixInverse},
	so running \Cref{lem:strongConnectivity} in parallel
	does not further affect their complexity.
}
\end{lemma}

By running this dynamic algorithm in parallel, 
we can detect if the graph is no longer strongly connected,
in which case our dynamic diameter result will simply return $\infty$.

\subsection{Unweighted Approximate Diameter (Proof of \Cref{thm:approxDiam_unweighted,thm:1.5approxDiam})}
\label{sub:largeDiameter}

The high-level idea for unweighted, $(1+\varepsilon)$-approximate diameters is very simple:
For diameter less than $n^s$, we can find the diameter by computing all-pairs-distances via
\Cref{thm:dynamicSmallApproximateDistancesQuery}.
If the diameter is larger than $n^s$,
then for a random hitting set the longest shortest path can be decomposed
into segments $s \to h_1 \to ... \to h_k \to t$,
where $h_i \in H$ are hitting-set nodes
and each segment has length at most $n^s \varepsilon$.
Thus we can find an approximate diameter
by looking only for the longest path between the nodes $H$.

\begin{lemma}\label{thm:diameterApproximationViaSampling}
Let $G$ be an $n$-node graph with integer edge weights from $\{1,2,...,W\}$.

Let $H = \{h_1,h_2,...\} \subset V$ be a random subset of size
$c 2n/(d\varepsilon^2) \ln n$. 
Define $G_H$ to be the graph on node set $H$ and edge set $\{ (u,v) \mid \dist_G(u,v) \le Wd \}$ with cost $c_{u,v} = \dist_G(u,v)$.

Then with probability at least $1-n^{2-c}$ we have:
\begin{align*}
\diam(G) 
&\le \diam(G_H) + Wd\varepsilon \\
&\le (1+\varepsilon) \diam(G)
\text{ if }\diam(G) \ge Wd.
\end{align*}
If the graph is undirected, we additionally also have
\begin{align*}
\radius(G) &\le \radius(G_H) + 0.5 Wd \varepsilon \\
&\le (1+\varepsilon) \radius(G)
\text{ if }\radius(G) \ge Wd
\end{align*}
\end{lemma}

We will prove \Cref{thm:diameterApproximationViaSampling} later in this section.
For now we use it to show how to maintain a $(1+\varepsilon)$-approximation of the diameter.

\begin{theorem}[Approximate Diameter, Unweighted (or with $W$), Directed, $n^2$]
\label{thm:approxDiam_unweighted}
Let $G$ be a directed graph with $n$ nodes and integer weights from $\{1,...,W\}$
and let $\ell$ be such that $W = n^\ell$.
Then for any $\varepsilon > 0$ there exists a Monte Carlo dynamic algorithm
that maintains $(1+\varepsilon)$-approximate diameter of $G$ in
$ \tilde{O}(
Wn^{\weightUnwUndAPSP}/\varepsilon
+ n^2 / \varepsilon^{1+\omega}
)$ update time.
The pre-processing requires $\tilde{O}(Wn^{\preUnwUndAPSP} / \varepsilon)$ time.
\end{theorem}

\begin{proof}
Let $0 \le s \le 1$ be some parameter.
We maintain the $(1+\varepsilon)$-approximate $n^s$-hops distances via
\Cref{thm:dynamicSmallApproximateDistancesQuery}
(by maintaining the distances upto $Wn^s = n^{s+\ell}$ for $n^\ell = W$).
After every update we query all distances via $I = J = V$.
Thus every update requires $\tilde{O}(n^{\omega(1,1,s+\ell)}/\varepsilon)$ time per update by choosing $\mu = 0$.

We also obtain a pair $P \subset V \times V$ of pairs with distance greater
than $Wn^s$ (i.e. all pairs where the returned distance is $\infty$). 
If this set is empty, then we can get the diameter by looking at 
the longest distance we computed. If the set is non-empty, then we know the
diameter must be larger than $Wn^s$.

We also run \Cref{lem:strongConnectivity} in parallel to check if the graph
is strongly connected. If it is not strongly connected, then we return $\infty$ as diameter.
Otherwise, if the graph is strongly connected, then we can use
\Cref{thm:diameterApproximationViaSampling} to compute an approximation of the diameter.
For this we sample a random set $H \subset V$ of size $\tilde{O}(n^{1-s} / \varepsilon)$
and construct the graph $G_H$ as in \Cref{thm:diameterApproximationViaSampling}
in $\tilde{O}(n^{2-2s}/\varepsilon^2)$ time.

We can compute the approximate diameter of $G_H$
by computing $(1+\varepsilon)$-approximate all-pairs-distances
in $\tilde{O}(n^{(1-s)\omega} / \varepsilon^{1+\omega})$ time using \Cref{lem:allPathsHittingSet}.
(Note that we only have approximate distances of $G$ when constructing $G_H$,
so technically we obtain a $(1+\varepsilon)^3$-approximation of $\diam(G)$,
but we can just choose $\varepsilon$ to be small enough.)

The update time is thus $\tilde{O}(
n^{\omega(1,1,s+\ell)}/\varepsilon
+ n^{(1-s)\omega} / \varepsilon^{1+\omega})$,
which can be bounded by $\tilde O(Wn^{\weightUnwUndAPSP} / \varepsilon + n^2 / \varepsilon^{\omega+1})$.
(For details, see the proof of \Cref{thm:undirectedApproximateDistances}, where the same term was bounded.)

Note that the algorithm works against non-oblivious adversaries,
because (w.h.p.) the result from \Cref{thm:dynamicSmallApproximateDistancesQuery}
does not leak any information about the random choices
and we can sample a new hitting set in \Cref{thm:diameterApproximationViaSampling}
for every update.
\end{proof}

We are left with proving \Cref{thm:diameterApproximationViaSampling}.

\begin{proof}[Proof of \Cref{thm:diameterApproximationViaSampling}]

The main idea is to decompose paths in $G$ via \Cref{lem:allPathsHittingSet} and then map them to paths in $G_H$ and bound their length.

Let $H$ be the random set of size $c2n/(d\varepsilon^2)$,
then with probability at least $1-n^{2-c}$, every shortest path (using at least $d$ hops)
can be decomposed into segments
$s \to h_1 \to h_2 \to ... \to h_k \to t$
where each $h_i \in H$ and the segments 
have at most $\varepsilon d/2$ hops.

Thus for $u,v \in H$ we have $\dist_G(s,t) = \dist_{G_H}(u,v)$ and for every $s,t \in V$ with $\dist_G(s,t) \ge Wd$, there exist $u,v \in H$ such that $\dist_{G_H}(u,v) \le \dist_G(s,t) \le \dist_G(u,v) + Wd\varepsilon = \dist_{G_H}(u,v) + Wd\varepsilon$.

\paragraph{Approximating the Diameter}
Let $s,t \in V$ be the pairs such that $\dist_G(s,t) = \diam(G)$ 
and the shortest $st$-path uses at least $d$ hops, 
and let $s_H,t_H \in H$ be the first and last node from $H$ along the shortest $st$-path. 
Then $\dist_G(s_H,t_H) \le \dist_G(s,t) \le \dist_G(s_H,t_H) + Wd\varepsilon$ and
\begin{align*}
\diam(G) \le& \dist_G(s_H,t_H) + Wd\varepsilon \\
=& \dist_{G_H}(s_H,t_H) + Wd\varepsilon \le \diam(G_H),\\
\diam(G_H) 
=& \max_{u,v \in H} \dist_{G_H}(u,v) 
= \max_{u,v \in H} \dist_G(u,v) \\
\le& \max_{u,v \in V} \dist_G(u,v) 
= \diam(G).
\end{align*}
So if $\diam(G) \ge Wd$, then $\diam(G_H) + Wd \varepsilon$ is a $(1+\varepsilon)$-approximation of $\diam(G)$.

\paragraph{Approximating the Radius}

For the radius we can use the same arguments as for the diameter:
\begin{align*}
\radius(G)
&= \min_{s \in V} \max_{t \in V} \dist_G(s,t)
\le \min_{s \in H} \max_{t \in V} \dist_G(s,t) \\
&\le \min_{s \in H} \max_{t \in H} \dist_G(s,t) + 0.5 Wd \varepsilon \\
&= \min_{s \in H} \max_{t \in H} \dist_{G_H}(s,t) + 0.5 Wd \varepsilon \\
&= \radius(G_H) + 0.5 W d \varepsilon
\end{align*}
If the graph is undirected we can also get a bound from the other direction as follows:
Let $s \in V$ be the node such that $\radius(G) = \max_{v \in V} \dist_G(s,v)$ and assume $\radius(G) \ge Wd$, then w.h.p. there is a $s_H \in H$ with $\dist_G(s_h, s) \le 0.5Wd\varepsilon$ and thus $\max_{v\in V} \dist_G(s_H, v) \le \max_{v \in V} \dist_G(s,v) + 0.5 Wd\varepsilon$.
This leads to the following bound:
\begin{align*}
\radius(G_H)
&= \min_{s \in H} \max_{t \in H} \dist_{G_H}(s,t) \\
&= \min_{s \in H} \max_{t \in H} \dist_G(s,t) \\
&\le \min_{s \in H} \max_{t \in V} \dist_G(s,t) \\
&\le \min_{s \in V} \max_{t \in V} \dist_G(s,t) + 0.5 Wd \varepsilon \\
&= \radius(G) + 0.5 Wd \varepsilon
\end{align*}
Thus for undirected graphs with $\radius(G) \ge Wd$ we have
$\radius(G) \le \radius(G_H) + 0.5 Wd \varepsilon \le \radius(G) + Wd \varepsilon \le (1+\varepsilon) \radius(G)$,
so $\radius(G_H) + 0.5 Wd \varepsilon$ is a $(1+\varepsilon)$-approximation of the radius.
\end{proof}

\subsection{Unweighted 1.5 Approximate Diameter}
\label{sub:approximateDiameter}

The following result is from \cite[Lemma 4]{RodittyW13}, 
where Roditty and V. Williams
present a static algorithm to compute a nearly $1.5$-approximation of the diameter in unweighted graphs 
in $\tilde{O}(m \sqrt{n})$ time.
The high level idea is to compute BFS trees 
originating at $\tilde{O}(\sqrt{n})$ uniformly at random chosen nodes $S \subset V$, 
then find the node $w$ furthest away from the nodes $S$ 
and compute another BFS tree for this distant node $v$ 
and its $\sqrt{n}$ nearest neighbors. 
The largest computed distance of all BFS trees 
then yields a nearly $1.5$-approximation of the diameter.

We simply simulate their algorithm by querying 
the single-source (and single-sink) distances 
via \Cref{thm:dynamicSmallApproximateDistancesQuery} 
instead of performing the breadth-first-searches.

The original proof of \cite[Lemma 4]{RodittyW13} assumes exact distances
(i.e. the distances obtained by running BFS),
but the result also holds when we are given $(1+\varepsilon)$-approximate distances instead,
in which case we obtain a nearly $(1.5+\varepsilon)$-approximate diameter.
\ifdefined\FOCSversion
The proof of the following \Cref{thm:approxDiameter1.5} 
can be found in the full version of the paper.
\else
The proof of the following \Cref{thm:approxDiameter1.5} is in \Cref{app:approximateDiameter},
because it is identical to \cite[Lemma 4]{RodittyW13},
except that we need to slightly adapt some inequalities 
to verify that the algorithm does indeed work 
when using $(1+\varepsilon)$-approximate distances.
\fi

\begin{theorem}[{Based on \cite[Lemma 4]{RodittyW13}}]
\label{thm:approxDiameter1.5}
Let $G$ be a directed unweighted graph with diameter $D = 3h+z$, where $h \ge 0$ and $z \in \{0,1,2\}$.
Let $0 < \varepsilon \le 1$ and
let $\tilde{D}$ be the approximate diameter returned by \Cref{alg:approximateDiameter}.

Then we have w.h.p. $\min\{(2-3\varepsilon)h+z, 2h+1\} \le \tilde{D} \le (1+\varepsilon) D$, which can also be written as
$(\frac{2}{3}-\varepsilon) D - \frac{1}{3} \le \tilde{D} \le (1+\varepsilon) D$.
\end{theorem}

\begin{algorithm}\label{alg:approximateDiameter}

Assume we are given a $(1+\varepsilon)$-approximate distance matrix $\tilde{d}$.
This matrix does not need to be given explicitly,
it is enough if we can query rows/columns,
so $\tilde{d}$ could be some distance oracle.\footnote{\label{foot:obliviousOracle}
The result of the queries must be fixed, 
i.e. it is not allowed to depend on the order in which we query the results.
More formally, we require that 
for every $k_1,k_2$ and 
$u_1,v_1,...,u_{k_1},v_{k_1},$ $u'_1,v'_1,...,u'_{k_2},v'_{k_2},$ $s, t \in V$, 
the two sequences $\tilde{d}(u_1,v_1),...,\tilde{d}(u_{k_1},v_{k_1}),\tilde{d}(s,t)$ 
and $\tilde{d}(u'_1,v'_1),...,\tilde{d}(u'_{k_2},v'_{k_2}),\tilde{d}(s,t)$ 
both return the same result for $\tilde{d}(s,t)$.
}
\begin{enumerate}
\item Let $S \subset V$ be a random set of vertices of size $\tilde{\Theta}(\sqrt{n})$.

\item Query $(1+\varepsilon)$-approximate distances from/to the nodes in $S$,
so $\tilde{d}(s,u)$ for all $s \in S, u \in V$.

\item Let $w \in V$ be the node with largest (approximate) distance to $S$,
i.e. for all $u \in V$ we have
$\min_{s \in S} \tilde{d}(w,s) \ge \min_{s \in S} \tilde{d}(u,s)$.

\item Query the distance from/to $w$,
i.e. we query $\tilde{d}(u,w)$ and $\tilde{d}(w,u)$ for every $u \in V$.

\item Let $W \subset V$ be a set of size $\sqrt{n}$ such that
$\tilde{d}(w,u) \le \tilde{d}(w,v)$ for all $u \in W$ and $v \in V \setminus W$.\\
Ties are broken by node index.
(The set $W$ are the $\sqrt{n}$ approximately closest nodes to $w$.)

\item $\tilde{D} := \max_{v \in V, u\in S \cup W \cup \{ w\}}
\{ \tilde{d}(u,v), \tilde{d}(v,u)\}$,
i.e. the largest of all so far computed distances.

\end{enumerate}

\end{algorithm}

\begin{proof}[Proof of \Cref{thm:1.5approxDiam}]
Let $0 \le s$ be some parameter.
We run \Cref{thm:dynamicSmallApproximateDistancesQuery}, which allows us to query $(1+\varepsilon)$-approximate distances upto $n^s$, and we also run \Cref{lem:strongConnectivity} in parallel to check, if the graph is strongly connected.

If it is not strongly connected, then we know the diameter is $\infty$. Otherwise we proceed by running \Cref{thm:approxDiameter1.5} as described in the next paragraph.

\paragraph{Case 1, Simulating \cite{RodittyW13}}

We use the approximate distances of 
\Cref{thm:dynamicSmallApproximateDistancesQuery} 
to run \Cref{thm:approxDiameter1.5} (\Cref{alg:approximateDiameter}).
Note that \Cref{thm:approxDiameter1.5} performs the following queries:
(i) $\tilde{O}(\sqrt{n})$ sources/sinks (set $S$)
(ii) one source/sink (node $w$)
(iii) $\sqrt{n}$ sources/sinks (set $W$).
The time required for all queries (i), (ii), (iii) together is
$\tilde{O}(n^{\omega(0.5,\mu+s,1)})$.

Note that \Cref{thm:dynamicSmallApproximateDistancesQuery} 
only maintains distances upto $n^s$.
If some distance of a queried pair is larger than $n^s$,
then we can detect that, but we do not receive the actual distance.
In such a case we know that the diameter must be larger than $n^s$,
so we abort \Cref{alg:approximateDiameter} and instead compute the diameter differently.

\paragraph{Case 2, $\diam(G) \ge n^s$}

If \Cref{alg:approximateDiameter} fails, 
because some computed distance is larger than $n^s$, 
then we also know that $\diam(G) \ge n^s$.
This means we can now use \Cref{thm:diameterApproximationViaSampling} instead.
Let $H$ be the random set of nodes that \Cref{thm:diameterApproximationViaSampling} uses,
then querying the distances of the pairs $H \times H$ requires
$\tilde{O}(	n^{\omega(1-s,\mu+s,1-s)} / \varepsilon^2)$ time, as $|H| = \tilde{O}(n^{1-s}/\varepsilon)$.

Hence the total time for running 
\Cref{thm:diameterApproximationViaSampling,lem:allPairsApproximateDistances} 
to approximate the diameter becomes
$\tilde{O}(
	n^{\omega(1-s,\mu+s,1-s)}/\varepsilon^2
	+n^{(1-s)\omega} / \varepsilon^{1+\omega})$ time.
	
\paragraph{Update time}

The update time complexity for approximating the diameter is:
\begin{align*}
\tilde{O}(&
	\underbrace{
		n^{\omega(1,s+\mu,1)-\mu}/\varepsilon
		+ u(n^s,n)
	}_{\text{\Cref{thm:dynamicSmallApproximateDistancesQuery}}}
	+ \underbrace{
		n^{\omega(0.5,\mu+s,1)}/\varepsilon)
	}_{\text{\Cref{alg:approximateDiameter}}}\\
	+& \underbrace{
		n^{\omega(1-s,\mu+s,1-s)}/\varepsilon^2
	}_{\text{pairs } H \times H}
	+ \underbrace{
		n^{(1-s)\omega} / \varepsilon^{1+\omega}
	}_{\text{\Cref{thm:diameterApproximationViaSampling,lem:allPairsApproximateDistances}}}
)
\end{align*}

\paragraph{Radius}

The algorithm works almost identically, if we want to maintain the diameter.
If during the last step of \Cref{alg:approximateDiameter} 
we set $\tilde{R} := 
\min_{v\in S \cup W \cup \{ w\}} 
\max_{u \in V} \{ \tilde{d}(u,v)\}$, 
then $\tilde{R}$ is a nearly $(1.5+\varepsilon)$-approximation 
\ifdefined\FOCSversion
of the radius (see the full version for a proof of this claim).
\else
of the radius (see \Cref{lem:approximateRadius}).
\fi

We again have to consider the two cases $\radius(G) > n^s$ 
and $\radius(G) \le n^s$.
\begin{itemize}
\item 
	By running \Cref{lem:strongConnectivity}, 
	we can detect if $\radius(G) = \infty$,
	because then the graph is not connected.
	If the graph is connected, we proceed to the next case:
\item 
	Let $r$ be the node with the property $\radius(G) = \max_{v \in V} \dist(r,v)$, 
	then for any $u,v \in V$ we have 
	$\dist(u,v) \le \dist(u,r) + \dist(r,v) \le 2 \radius(G)$. 
	Hence it is enough to run \Cref{thm:dynamicSmallApproximateDistancesQuery} 
	for distances upto $2n^s$ 
	and compute $\tilde{R}$ 
	\ifdefined\FOCSversion
	via the adaption to \Cref{alg:approximateDiameter}.
	\else
	via the adaption to \Cref{alg:approximateDiameter} 
	(\Cref{lem:approximateRadius}).
	\fi
	If during some distance query to 
	\Cref{thm:dynamicSmallApproximateDistancesQuery} 
	we realize that the distance is larger than $2n^s$, 
	then we know the radius is larger than $n^s$. 
	We then cancel \Cref{alg:approximateDiameter} 
	and go to the next case.
\item $n^s \le \radius(G) < \infty$: 
	This case is identical to the case 
	where the diameter is more than $n^s$. 
	We simply use 
	\Cref{thm:diameterApproximationViaSampling,lem:allPairsApproximateDistances} 
	to obtain an approximation of the radius.
\end{itemize}

Note that \Cref{thm:dynamicSmallApproximateDistancesQuery}
works against non-oblivious adversaries,
so \Cref{thm:1.5approxDiam} works against
non-oblivious adversaries as well,
because we sample a new random hitting-set $S$ in \Cref{alg:approximateDiameter} after every update.
\end{proof}

We can also maintain nearly $(5/3+\varepsilon)$-approximate eccentricities,
by dynamically maintaining $(1+\varepsilon)$-approximate distances
and simulating a variant of the static algorithm for approximating eccentricities of \cite{ChechikLRSTW14}.
The static algorithm would usually perform BFS searches,
but instead we query the distances via our dynamic algorithm.

\begin{theorem}\label{thm:approximateEccentricities}
Let $G$ be an unweighted, undirected $n$-node graph.
Then for any $0 < \varepsilon$, $0 \le \mu \le 1$, $0 \le s < 0.5$,
there exists a Monte Carlo dynamic algorithm that maintains nearly $(3/5+\varepsilon)$-approximate eccentricities $\tilde{ecc}( \cdot )$, such that for every $v$:
$$
\left(\frac{3}{5}-\varepsilon\right) ecc(v) - 4/7 \le \tilde{ecc}(v) \le ecc(v)
$$
The update time is
\begin{align*}
\tilde{O}(&
	n^{\omega(1,s+\mu,1)-\mu} / \varepsilon
	+ u(n^s, n)) / \varepsilon \\
	&+ n^{\omega(1,\mu+s,1-s)} / \varepsilon^2
	+ n^{(1-s)\cdot\omega} / \varepsilon^{\omega+1}
),
\end{align*}
and the pre-processing requires $\tilde{O}(n^{\omega+s})$ time
(both complexities are the same as in \Cref{thm:undirectedApproximateDistancesQuery}).
For current $\omega$ the pre-processing and update time are
$\tilde{O}(n^{\preUnwUndAPSPquery})$ and
$\tilde{O}(n^{\updateUnwUndAPSPquery} / \varepsilon^{3.373})$ respectively
with $s \approx \sUnwUndAPSPquery$ and $\mu \approx \muUnwUndAPSPquery$.
\end{theorem}

\begin{proof}
We simulate a variant of the static algorithm for approximating eccentricities 
\ifdefined\FOCSversion
of \cite{ChechikLRSTW14} as defined in the full version of our paper.
\else
of \cite{ChechikLRSTW14} as defined in \Cref{lem:approximateEccentricities}.
\fi
For that we need to query the distances of $\tilde{O}(\sqrt{n})$ source nodes to every other node.
As the graph is undirected this can be done using the dynamic algorithm of \Cref{thm:undirectedApproximateDistancesQuery}.
The query time for this many sources is
$
\tilde{O}(n^{\omega(0.5,s+\mu,1)})
$,
which together with the update time of \Cref{thm:undirectedApproximateDistancesQuery} results in our dynamic eccentricities algorithm having the update time
\begin{align*}
\tilde{O}(&
	n^{\omega(1,s+\mu,1)-\mu} / \varepsilon
	+ u(n^s, n)) / \varepsilon \\
	+& n^{\omega(1,\mu+s,1-s)} / \varepsilon^2
	+ n^{(1-s)\cdot\omega} / \varepsilon^{\omega+1}
	+ n^{\omega(0.5,s+\mu,1)}
).
\end{align*}
Note that $u(n^s, n)) > n^{1.5+s}$, even if $\omega = 2$,
hence we can only obtain a subquadratic algorithm for $s < 0.5$.
With this observation the $n^{\omega(0.5,s+\mu,1)}$ term
is subsumed by $n^{\omega(1,\mu+s,1-s)} / \varepsilon^2$.
So the update time of \Cref{thm:approximateEccentricities} is the same as \Cref{thm:undirectedApproximateDistancesQuery}.

Note that \Cref{thm:undirectedApproximateDistancesQuery}
works against non-oblivious adversaries,
so \Cref{thm:approximateEccentricities} works against
\ifdefined\FOCSversion
non-oblivious adversaries as well.
\else
non-oblivious adversaries as well,
because we sample a new random hitting-set $S$ in \Cref{lem:approximateEccentricities} after every update.
\fi
\end{proof}

%% file: open.tex
\section{Open Problems}\label{sec:open}

An obvious open problem is to improve our bounds and proving matching conditional lower bounds. 
Improving our bounds with amortization should already be interesting (except for APSP).
%
Another intriguing question is 
whether fast matrix multiplication is really needed 
to get the bounds we achieve here. 
Note that a conditional lower bound of Abboud and V. Williams \cite{AbboudW14} 
suggests that this is the case for SSSP on directed graphs\footnote{%
In particular, under the Boolean Matrix Multiplication (BMM) conjecture, 
there is no ``combinatorial'' algorithm with $n^{3-\epsilon}$ preprocessing time 
and $n^{2-\epsilon}$ update time 
even for $st$-Reachability (maintaining the reachability between two nodes)}. 
We are not aware of similar lower bounds for other problems. 

Whether the bounds similar to ours hold without the $\epsilon$ term remains open (e.g. {\em exact} SSSP and APSP). 
A subquadratic update time for exact unweighted SSSP 
and exact weighted $st$-shortest paths will be very interesting. 
A subquadratic update time for exact weighted SSSP will be surprising, 
since such bound has already been ruled out for algorithms 
with $n^{3-\epsilon}$ proprocessing time \cite{AbboudW14}. 
Of course, showing $O(n^2)$ worst-case update time for maintaining APSP exactly remains a major open problem.

Finally, note that our algorithms can only maintain distances 
but cannot report the corresponding paths. 
Supporting such operation is interesting, 
especially doing for APSP in the same time complexity 
as Demetrescu and Italiano's algorithm \cite{DemetrescuI04}.

%% file: acknowledgement.tex
\section*{Acknowledgment}

The authors would like to thank André Nusser for detecting some errors in the previous version.

This project has received funding from the European Research Council (ERC) under the European
Unions Horizon 2020 research and innovation programme under grant agreement No 715672. Danupon
Nanongkai was also partially supported by the Swedish Research Council (Reg. No. 2015-04659).

%% file: appendix_worst_case.tex
\section{Worst-Case Standard Technique}
\label{app:worstCase}

\begin{theorem}\label{thm:standardTechnique}
Let $\mathcal{A}$ be a dynamic algorithm with reset time $O(r)$, update time $O(u(t))$ and query time $O(q(t))$, where $t$ is the number of past updates, since the last reset or initialization.

For every $\mu \in \N$ there is an algorithm $\mathcal{W}$ with worst-case update complexity $O(u(\mu) + r / \mu)$ and query complexity $O(q(\mu))$.
\end{theorem}

If we were interested in amortized complexity, \Cref{thm:standardTechnique} would be trivial by simply reset the algorithm after every $\mu$ updates.

\begin{proof}[Proof of \Cref{thm:standardTechnique}]

For simplicity assume $\mu$ is a multiple of 4.
We maintain two copies of $\mathcal{A}$ in parallel,
where each copy will have the following life-cycle:

\begin{enumerate}

\item For the next $\mu / 4$ updates, the copy performs its reset operation.
This means for every of the next $\mu / 4$ updates the algorithm will perform
$O(r) / (\mu / 4)) = O(r / \mu)$ operations of the reset routine.

\item The reset routine was started $\mu / 4$ updates in the past,
so the last $\mu / 4$ updates were not yet applied to the copy.
To fix this we will always perform two queued updates for each of the next $\mu / 4$ updates.
This means we have $2 \cdot O(u(\mu/2)) = O(u(\mu))$ cost per update
and after $\mu / 4$ updates the copy has caught up with all queued updates.

\item For the next $\mu / 2$ rounds the copy can perform updates normally in $O(q(\mu))$ time and is able to answer queries in $O(q(\mu))$ time.

\item The copy now has received a total of $\mu$ updates and needs to reset again, so jump back to step 1.

\end{enumerate}

Note that during this cycle, each copy is alternating between being unavailable
(resetting + catching up)
and being available
(performing current updates/answering current queries)
for a sequence of $\mu / 2$ rounds each.
Thus we simply need two copies of the algorithm that are phase shifted,
then one copy is always available for the queries.

Both copies are initialized at the same time,
and the phase-shift can easily be obtained by simply resetting the first copy
directly after the initialization.
So copy 1 starts with phase 1 while copy 2 starts with phase 3.

\end{proof}

%% file: appendix_reduction.tex
\section{Reduction from Distances to Polynomial Matrix Inverse}
\label{app:reduction:distanceToInverse}

Obtaining the distances in graphs via 
computing the inverse (or adjoint) of a polynomial matrix 
is a commonly used technique 
\cite{Sankowski05,Sankowski05ESA,CyganGS15,BrandNS18,BrandS18}.
Usually the reduction yields the exact distances. 
Here we prove a variant of that reduction 
that allows us to obtain upper bounds on the distance, 
which can then be used to obtain approximate distances.

We first repeat the reduction
\Cref{lem:distanceInverseReduction} 
as given in \Cref{sec:overview}.

\begin{lemma}

Let $G$ be a directed graph 
with positive integer edge weights $(c_{u,v})_{(u,v)\in E}$
and let $d$ be some positive integer.
We define $A(G) \in (\F[X]/ \langle X^h \rangle)^{n \times n}$,
such that $A(G)_{u,v} = a_{u,v} X^{c_{u,v}}$ for each edge $(u,v)$,
$A(G)_{v,v} = a_{v,v} X$ for every $v \in V$, and $A_{u,v} = 0$ otherwise.
Here each $a_{u,v}$ is an independent 
and uniformly at random chosen element from $\F$.
\footnote{
	Note that for $c_{u,v} \ge d$ we have $A_{u,v} = 0$,
	as if the edge would not exist.
}
	
Then the matrix $M = \I - A(G)$ is invertible and with probability
at least $1 - dn^2/|F|$ the following property holds:

For every $u,v \in V$ and $0 \le d < h$
the entry $(M^{-1})^{[d]}_{u,v} \neq 0$ 
if and only if 
$\dist(u,v) \le d$.

\end{lemma}

\begin{proof}

For now assume $A(G)_{v,v} = 0$, which is a slightly simpler case.
Let $W_{s,t}^k$ be the set of walks from $s$ to $t$ in $G$ using exactly $k$ steps,
and given a walk $w = (v_0,v_1,...,v_k)$ define $c_w$ to be the distance of the walk, 
so $c_w := \sum_{i=1}^{k} c_{v_{i-1},v_i}$.
Further define $a_w := \prod_{i=1}^{k} a_{v_{i-1}, v_i}$,
then
\ifdefined\FOCSversion
\begin{align*}
A(G)^k_{s,t} 
&= \sum_{(v_0,v_1,...,v_k) \in W^k_{s,t}} \prod_{i=1}^{k} a_{v_{i-1},v_i} X^{c_{v_{i-1},v_i}} \\
&= \sum_{w \in W^k_{s,t}} a_w X^{c_w}
.
\end{align*}
\else
$$A(G)^k_{s,t} 
= \sum_{(v_0,v_1,...,v_k) \in W^k_{s,t}} \prod_{i=1}^{k} a_{v_{i-1},v_i} X^{c_{v_{i-1},v_i}}
= \sum_{w \in W^k_{s,t}} a_w X^{c_w}
.$$
\fi

Thus the coefficient of $X^d$ of $\sum_{i=0}^{h-1} (A(G)^i)_{s,t}$ is the sum of all $a_w$, 
where $w$ are $st$-walks of distance $d$, using less than $h$ steps.
This coefficient can be considered a polynomial $P$ in $(a_{u,v})_{(u,v)\in E}$ of degree at most $h$.
This polynomial $p$ is the zero-polynomial, if and only if no such $st$-walk exists.
We can now apply the Schwartz-Zippel Lemma \cite{Schwartz80,Zippel79}:
When evaluating the polynomial $P$ of degree at most $h$, 
where the input $(a_{u,v})_{u,v\in E}$ 
are independent and uniformly at random chosen elements from $\F$,
then with probability at most $h/|\F|$ the polynomial evaluates to zero, 
even though it is not the zero-polynomial.

Thus via union bound, with probability at least $1-hn^{2}/|\F|$, 
we have that for every $s,t \in V$ and $0 \le d < h$ the entry
$(\I - A(G))^{[d]}_{s,t} = \sum_{i=0}^{h-1} (A(G)^i)^{[d]}_{s,t}$ is non-zero, 
if and only if there exists a walk from $u$ to $v$ of distance exactly $d$.

Now note that having extra entries $A(G)_{v,v} = a_{v,v} X$ 
is equivalent to adding self-loops of cost 1 to every node.
By adding these self-loops there always exists an $st$-walk of distance $d$ for any $d \ge \dist(s,t)$.

\end{proof}

Note that our reduction uses walks, i.e. nodes and edges are allowed to be used more than once. 
In \cite{Koutis08,Williams09} a similar reduction was used to obtain simple paths of length $k$ (where nodes/edges are used only once per path). 
There the matrix $A(G) \in \F[(X_v)_{v \in V}]$ was multivariate 
and $A(G)_{u,v} = X_u$ for every edge $(u,v)$. 
The matrix $A(G)^k$ can then be used to find $k$-paths.

%% file: appendix_omega.tex
\section{Bounding $\omega( \cdot, \cdot, \cdot)$}
\label{app:omega}

We can use the upper bounds from \cite{GallU18} for $\omega(1,1,k)$ 
to get upper bounds on any $\omega(a,b,c)$ via the following routine. 
Without loss of generality assume $a \ge b \ge c$, 
otherwise rename/reorder the variables.
\begin{algorithmic}[1]
\IF{$a = b =c$}
\RETURN $a \cdot \omega$
\ENDIF
\IF{$a=b$}
\RETURN $a \cdot \omega(1,1,c/a)$
\ENDIF
\IF{$b=c$}
\RETURN $b \cdot \omega(1,1,a/b)$ (Note that \cite{GallU18} also proved bounds for $\omega(1,1,k)$ where $k > 1$.)
\ENDIF
\RETURN $\min\{(a-b) + b \cdot \omega(1,1,c/b), (b-c) + c \cdot \omega(1,1,a/c)\}$
\end{algorithmic}

Using this bound on $\omega(a,b,c)$, one can use an optimization program 
to balance the terms of the complexities of our algorithms.%
\footnote{
	An online version for bounds on $\omega(a,b,c)$ is available at \url{https://people.kth.se/~janvdb/matrix.html}.
	The optimization program we used was implemented in python and is available at \url{https://people.kth.se/~janvdb/shortestpaths_complexity.zip}.
}
For instance for the complexity $O(n^{1+\nu}+n^{\omega(1,1,\nu)-\nu})$ 
of \Cref{lem:dynamicMatrixInverse} for $s=0$,
we have to compute $\min_{0 \le \varepsilon \le 1} \max \{ 1+\nu, \omega(1,1,\nu)-\nu\}$ 
in order to get the smallest possible update complexity.

%% file: appendix_hittingset.tex
\section{Hitting-Set Arguments}
\label{app:hittingSet}

\begin{lemma}\label{lem:singlePathHittingSet}
Let $G$ be a graph with $n$ nodes and let $H \subset V$ be a uniformly chosen random subset of size $c \frac{n}{d} \ln n$.

Then for any path $P = (v_1,...,v_l)$ with $l \ge d$
there exists $1 \le i \le d$ such that $v_i \in H$ with probability at least $1-n^{-c}$.
\end{lemma}

\begin{proof}
The probability of there being no $1 \le i \le d$ with $v_i \in H$ is at most
$(1-|H|/n)^{d} \le e^{-|H|d/n} = 1-n^{-c}$.
\end{proof}

\begin{proof}[Proof of \Cref{lem:allPathsHittingSet}]
From \Cref{lem:singlePathHittingSet} via union bound over all pairs $(u,v) \in V^2$, we obtain the following statement:

With probability at least $1-n^{2-c}$, we have that for every shortest path
$P = (v_1,...,v_l)$ with $l \ge d$
there exists $1 \le i \le d$ such that $v_i \in H$.

Since every segment of a shortest path is itself a shortest path,
we also have that, given a shortest path $P = (s,...,t)$ with $l \ge d$,
the path be decomposed into segments
$s \to h_1 \to h_2 \to ... \to h_k \to t$,
where $h_i \in H$ for every $i=1,...,k$
and each segment uses at most $d$ hops.

\end{proof}

\paragraph{Non-oblivious adversaries}

All algorithms that build on top of the results of \Cref{sec:smallHop}, use hitting-set arguments.
The answer of queries (e.g. the quality of the approximation) 
could leak information about which nodes were sampled as hitting-sets.
However, we can simply sample a new hitting set after every edge update. 
This way adversaries can not break the algorithm 
by choosing their updates according to the output of queries.

%% file: appendix_approximate_diameter.tex
\section{Approximate Diameter, Radius and Eccentricities}
\label{app:approximateDiameter}

In this section we prove
that we can use the static diameter algorithm from \cite{RodittyW13}
even if the given distances are only $(1+\varepsilon)$-approximate.
The proof for diameter is identical to \cite{RodittyW13},
we just verify that it does not break when given approximate distances.
We also verify that our modification to the eccentricities algorithm from \cite{ChechikLRSTW14}
works with approximate distances.

For a subset $U \subset V$ we write
$\dist(v, U) := \min_{u \in U} \dist(v,u)$
and we call this the distance of node $v$ to set $S$.

\begin{lemma}\label{lem:neighborSampling}

Let $\tilde{d}(u,v)$ be $(1+\varepsilon)$-approximate
distances between any $u$ an $v$.
We write $N(v)$ for the set of size $\sqrt{n}$, such that
$\tilde{d}(v, u) \le \tilde{d}(v, w)$
for all $u \in N(v), w \in V \setminus N(v)$.
We will break ties by choosing nodes with smallest index first.

If choose set $S \subset V$ uniformly at random with size
$c \sqrt{n} \log n = \tilde{O}(\sqrt{n})$ for some constant $c$,
then with probability at least $n^{1-c}$ we have $N(v) \cap S \neq \emptyset$ for all $v \in V$.

\end{lemma}

\begin{proof}

Fix some node $v \in V$, then the probability of not a single node of $S$ hitting $N(v)$ is
$(1-|N(v)|/n)(1-|N(v)|/(n-1))...(1-|N(v)|/(n-|S|))
\le (1-1/\sqrt{n})^{c\sqrt{n}\log n}
\le n^{-c}$

Via union bound the probability that there exists at least one $v \in V$ such that $N(v) \cap S = \emptyset$ is at most $n^{1-c}$.
Hence with probability at least $1-n^{1-c}$ we have that for all $v \in V$ $S \cap N(v) \neq \emptyset$.

\end{proof}

Note that in \Cref{lem:neighborSampling} we are first given
the approximate distances $\tilde{d}(v,u)$ for all $u,v \in V$ and thus (implicitly) the sets $N(v)$,
and then afterward we sample the set $S$ independently of $\tilde{d}(v,u)$ and $N(v)$.

In \Cref{alg:approximateDiameter} the order is flipped:
We first sample $S$ and then compute/query $\tilde{d}(w,u)$ for every $u \in V$
and construct $W = N(w)$.

We can still use \Cref{lem:neighborSampling}
to prove the correctness of \Cref{alg:approximateDiameter},
because the sets $N(v)$ are already fixed before sampling $S$, 
when satisfying the requirement of footnote \ref{foot:obliviousOracle}.

\begin{proof}[Proof of \Cref{thm:approxDiameter1.5}]
Let $D = 3h+z$ be the diameter of $G$ where $z \in \{0,1,2\}$.
The claim of \Cref{thm:approxDiameter1.5} is that the output $\tilde{D}$ of \Cref{alg:approximateDiameter} satisfies
$\min\{ (2-3\varepsilon)h+z, 2h+z\} \le \tilde{D} \le (1+\varepsilon) D$.

Let $a,b \in V$ be such that $\dist(a,b) = D$.
First notice that the algorithm always
returns the approximate depth of some shortest paths tree and hence
$\tilde{D} \le (1+\varepsilon) D$.

Consider the cases $\dist(w,S) \le h$ and $\dist(w,S) > h$:
\paragraph{Case $\dist(w,S) \le h(1+\varepsilon)$:}
If $\dist(w,S) \le h(1+\varepsilon)$ then
$\dist(a, S) \le \tilde{d}(s,S) \le
\le \tilde{d}(w,S)
\le (1+\varepsilon)\dist(w,S)
\le (1+\varepsilon)^2 h$.

As $ 3h+z = \dist(a,b)
\le \min_{s \in S} \dist(a,s)+\dist(s,b)
\le \dist(a,S) + \tilde{D}
\le (1+\varepsilon)^2 h + \tilde{D}$,
we have $\tilde{D} \ge (3-(1+\varepsilon)^2)h+z \ge (2-3\varepsilon)h+z$.
The last inequality only holds when assuming $\varepsilon \le 1$.

\paragraph{Case $\dist(w,S) > h(1+\varepsilon)$:}

We will first argue that $\max_{v \in W} \tilde{d}(w,v) > h(1+\varepsilon)$:
Assume $\max_{v \in W} \tilde{d}(w,v) \le h(1+\varepsilon)$.
W.h.p. there is at least one $s \in S \cap W$,
because of \Cref{lem:neighborSampling} and $W = N(w)$, so
$h(1+\varepsilon) < 
\dist(w,S) 
\le \tilde{d}(w,S) 
\le \tilde{d}(w,s) 
\le \max_{v \in W} \tilde{d}(w,v) 
\le h(1+\varepsilon)$,
which is a contradiction.
Thus $\max_{v \in W} \tilde{d}(w,v) > h(1+\varepsilon)$.

This implies that all nodes $v$ of distance at most $h$ from $w$ must be included in the set $W$,
because $\dist(w,v) \le h$ implies $\tilde{d}(w,v) \le h(1+\varepsilon)$ 
and because $W$ are the (approximately) closest neighbors of $w$.
This also implies that there exists $w' \in W$ on the shortest path from $w$ to $b$,
such that $\dist(w,w') = h$.

Note that we can also assume $ecc(w) < 2h+z$,
because otherwise we already obtain $\tilde{D} \ge 2h+z$ via the maximum distance from $w$.
It follows that $\dist(w,b) < 2h + z$, which means
$\dist(w',b)
= \dist(w,b) - \dist(w,w')
< 2h + z - h
= h + z$.
Since $\dist(a,b) = 3h + z$ this implies $2h + 1 \le \dist(a,w') \le \tilde{d}(a,w') \le \tilde{D}$.
\end{proof}

\begin{lemma}\label{lem:approximateRadius}
Consider \Cref{alg:approximateDiameter} when run on an undirected graph, but instead of the maximum depth, we return the minimum depth: $\min_{s \in S \cup W \cup \{w\}} \max_{v \in V} \tilde{d}(s,v)$.

Let $\tilde{R}$ be the returned value and $R$ be the radius of the graph, then w.h.p.
$$
R \le \tilde{R} \le ((1.5+\varepsilon) R + 2/3)(1+\varepsilon)
$$
\end{lemma}

\begin{proof}
We have $\tilde{R} \ge R$, because
\ifdefined\FOCSversion
\begin{align*}
R &= \min_{u \in V} \max_{v \in V} \dist(u,v)
\le \min_{u \in V} \max_{v, \in V} \tilde{d}(u,v) \\
&\le \max_{v \in V} \tilde{d}(s,v) \text{ for all } s \in V.
\end{align*}
\else
$$R
= \min_{u \in V} \max_{v \in V} \dist(u,v)
\le \min_{u \in V} \max_{v, \in V} \tilde{d}(u,v)
\le \max_{v \in V} \tilde{d}(s,v) \text{ for all } s \in V.$$
\fi
So we are left with proving $\tilde{R} \le (1.5 R + 2/3)(1+\varepsilon)$.

Let $x$ be the node such that $R = ecc(x)$,
and for any $v \in S \cup W \cup \{ w\}$ let $t_v \in V$ be such that $\dist(v,t_v) = ecc(v)$.
Note that for any $v \in S \cup W \cup \{ w\}$ we have that
\ifdefined\FOCSversion
\begin{align*}
\tilde{R}/(1+\varepsilon) &\le \max_{u \in V} \tilde{d}(v,u)/(1+\varepsilon)
\le ecc(v) \\
&\le \dist(v,x) + \dist(x,v_t) \\
&\le \dist(v,x) + ecc(x)
= \dist(v,x) + R,
\end{align*}
\else
$$\tilde{R}/(1+\varepsilon) \le \max_{u \in V} \tilde{d}(v,u)/(1+\varepsilon)
\le ecc(v) 
\le \dist(v,x) + \dist(x,v_t) 
\le \dist(v,x) + ecc(x)
= \dist(v,x) + R,$$
\fi
so whenever some $v \in S \cup W \cup \{ w\}$ has the property $\dist(v,x) \le 0.5 R + 2/3$, then $\tilde{R} \le (1.5 R + c)(1+\varepsilon)$ and we are done.
Thus we can assume $\dist(x, S) > 0.5 R + 2/3$ and by definition of $w$ we have
$\tilde{d}(w, S) \ge \tilde{d}(x, S) \ge \dist(x,S) > 0.5 R + 2/3$.

Since (w.h.p.) $S \cap W \neq \emptyset$, we have some node $s \in W$ with $\tilde{d}(w,s) > 0.5 R + 2/3$.
Thus $u \in V$ with $\tilde{d}(w, u) \le 0.5 R + c$ must be in $W$ as well by definition of $W$.
This then also implies that nodes $u \in V$ with $\dist(w,u) \le (0.5 R + 2/3) / (1+\varepsilon)$ must be in $W$.

Hence we can choose some $w' \in W$ on the shortest $wx$-path with $\dist(w, w') = \lfloor (0.5 R + 2/3) / (1+\varepsilon) \rfloor$.
We now know
\ifdefined\FOCSversion
\begin{align*}
R
&= ecc(x)
> \dist(w,x) 
= \dist(w,w') + \dist(w', x) \\
&= \lfloor (0.5 R + 2/3)/(1+\varepsilon) \rfloor + \dist(w', x),
\end{align*}
so
\begin{align*}
\dist(w', x) 
&< R - \lfloor (0.5 R + 2/3)/(1+\varepsilon) \rfloor) \\
&\le R - 0.5 R / (1+\varepsilon) + (1-2/3/(1+\varepsilon)) \\
&\le (0.5+\varepsilon) R + 2/3)
\end{align*}
\else
$$
R
= ecc(x)
> \dist(w,x) 
= \dist(w,w') + \dist(w', x)
= \lfloor (0.5 R + 2/3)/(1+\varepsilon) \rfloor + \dist(w', x),$$
so
$$\dist(w', x) 
< R - \lfloor (0.5 R + 2/3)/(1+\varepsilon) \rfloor)
\le R - 0.5 R / (1+\varepsilon) + (1-2/3/(1+\varepsilon))
\le (0.5+\varepsilon) R + 2/3)
$$
\fi
Here the last inequality uses $\varepsilon \le 1$.
As $w' \in W$, we have
\ifdefined\FOCSversion
\begin{align*}
\tilde{R} 
&\le \tilde{ecc}(w')
\le ecc(w') (1+\varepsilon) \\
&\le (\dist(w',x) + ecc(x))(1+\varepsilon) \\
&= ((1.5+\varepsilon) R + 2/3) (1+\varepsilon).
\end{align*}
\else
$$\tilde{R} 
\le \tilde{ecc}(w')
\le ecc(w') (1+\varepsilon)
\le (\dist(w',x) + ecc(x))(1+\varepsilon)
= ((1.5+\varepsilon) R + 2/3) (1+\varepsilon).$$
\fi
\end{proof}

The following approximate eccentricities algorithm is based on \cite{ChechikLRSTW14}.
The algorithm in \cite{ChechikLRSTW14} yields a true $3/5$-approximation, 
but needs to iterate over all edges.
As we are interested in subquadratic results, 
we verify here that skipping this iteration over all edges does not break the algorithm,
however, we do incur an additive error.

\begin{lemma}\label{lem:approximateEccentricities}
Consider \Cref{alg:approximateDiameter} when run on an undirected graph, but instead of the maximum depth, we return the following $\tilde{ecc}$:
\ifdefined\FOCSversion
\begin{align*}
\tilde{ecc}(s) &:= \max_{v \in V} \tilde{d}(s,v)
\text{ for } s \in S \cup W \cup \{w\} \\
\tilde{ecc}(v) &:= \max_{s \in S \cup W \cup \{w \}} \max \{ \tilde{d}(s,v), \tilde{ecc}(s) - \tilde{d}(s,v) \} \\
&\text{ for all other } v \in V
\end{align*}
\else
\begin{align*}
\tilde{ecc}(s) &:= \max_{v \in V} \tilde{d}(s,v)
&\text{ for } s \in S \cup W \cup \{w\} \\
\tilde{ecc}(v) &:= \max_{s \in S \cup W \cup \{w \}} \max \{ \tilde{d}(s,v), \tilde{ecc}(s) - \tilde{d}(s,v) \} 
&\text{ for all other } v \in V
\end{align*}
\fi

Then w.h.p. for all $v \in V$ we have
$$
\frac{3-6\varepsilon}{5} ecc(v) - 4/7 \le \tilde{ecc}(v) \le ecc(v)
$$
\end{lemma}

\begin{proof}
We start by proving $ecc(v) \le \tilde{ecc}(v)$:
Let $t_s$ be the node such that $\dist(s,t_s) = ecc(s)$,
then $ecc(s) = \dist(s,t_s) \le \dist(s,v) + \dist(v, t_s) \le \dist(s,v) + ecc(v)$.
Thus we have $ecc(v) \ge ecc(s) - \dist(s,v)$.
For our approximation this means:
\ifdefined\FOCSversion
\begin{align*}
\tilde{ecc}(s) - \tilde{d}(s,v) 
&\le (1+\varepsilon) ecc(s) - \dist(s,v) \\
&\le ecc(v) + \varepsilon ecc(s)  \\
&\le ecc(v) + \varepsilon (\dist(s,v) + ecc(v)) \\
&\le (1+2\varepsilon) ecc(v)
\end{align*}
\else
\begin{align*}
\tilde{ecc}(s) - \tilde{d}(s,v) 
&\le (1+\varepsilon) ecc(s) - \dist(s,v) 
\le ecc(v) + \varepsilon ecc(s)  \\
&\le ecc(v) + \varepsilon (\dist(s,v) + ecc(v)) 
\le (1+2\varepsilon) ecc(v)
\end{align*}
\fi

Thus we obtain the following upper bound for $\tilde{ecc}(v)$: 
\ifdefined\FOCSversion
\begin{align*}
\tilde{ecc}(v)
&= \max_{s \in S \cup W \cup \{w \}} \max \{ \tilde{d}(s,v), \tilde{ecc}(s) - \tilde{d}(s,v) \}\\
&\le (1+2\varepsilon) \max_{s \in S \cup W \cup \{w \}} \max \{ \dist(s,v), \\
&\hspace{60pt} ecc(s) - \dist(s,v) \}\\
&\le (1+2\varepsilon) ecc(v)
\end{align*}
\else
\begin{align*}
\tilde{ecc}(v)
&= \max_{s \in S \cup W \cup \{w \}} \max \{ \tilde{d}(s,v), \tilde{ecc}(s) - \tilde{d}(s,v) \}\\
&\le (1+2\varepsilon) \max_{s \in S \cup W \cup \{w \}} \max \{ \dist(s,v), ecc(s) - \dist(s,v) \}\\
&\le (1+2\varepsilon) ecc(v)
\end{align*}
\fi

\paragraph{Lower Bound}
For the lower bound on $\tilde{ecc}(v)$, note that for any $s \in S \cup W \cup \{w\}$ we have
\ifdefined\FOCSversion
\begin{align*}
\tilde{ecc}(v)
&\ge \tilde{ecc}(s) - \tilde{d}(s,v)
\ge ecc(s) - (1+\varepsilon) \dist(s,v) \\
&\ge ecc(v) - (2+\varepsilon) \dist(s,v) \\
&\ge ecc(v) - 2(1+\varepsilon) \dist(s,v).
\end{align*}
\else
\begin{align*}
\tilde{ecc}(v)
&\ge \tilde{ecc}(s) - \tilde{d}(s,v)
\ge ecc(s) - (1+\varepsilon) \dist(s,v)\\
&\ge ecc(v) - (2+\varepsilon) \dist(s,v)
\ge ecc(v) - 2(1+\varepsilon) \dist(s,v).
\end{align*}
\fi
So in order to obtain $\tilde{ecc}(v) \ge (3-\varepsilon)/5 ecc(v) - 4/7$,
we want to prove that there is a $s \in S \cup W \cup \{w\}$ for which $\dist(s,v) \approx 1/5 ecc(v)$.

If $\tilde{d}(w,v) \ge 3/5 ecc(v) - 4/7$, then we are done,
so assume $\tilde{d}(w,v) > 3/5 ecc(v) - 4/7$.
Let $t$ be the node such that $\dist(v,t) = ecc(v)$,
then if $\dist(s,t) \le 2/5 ecc(v) + 4/7$ for some $s \in S \cup W \cup \{w\}$,
we would have $\dist(v,s) \ge 3/5 - 4/7$ and we would be done,
so assume $\dist(t,S) > 2/5 ecc(v) + 4/7$.

Thus we have $\tilde{d}(w,S) \ge \tilde{d}(t,S) \ge \dist(t,S) > 2/5 ecc(v) + 4/7$.
Note that $W \cap S \neq \emptyset$, so all nodes $u \in V$ with
$\tilde{d}(w,u) \le 2/5 ecc(v) + 4/7$ must also be in the set $W$.
This also means that all nodes $u \in V$ with
$\dist(w,u) \le (2/5 ecc(v) + 4/7) / (1+\varepsilon)$ must be in $W$.

Let $w'$ be the node in $W$ along the shortest $vw$-path with
$\dist(w',w) = \lfloor (2/5 ecc(v) + 4/7) / (1+\varepsilon) \rfloor$.
Then we have
\ifdefined\FOCSversion
\begin{align*}
&~3/5 ecc(v) - 4/7
> \dist(w,v) \\
=&~ \dist(w,w') + \dist(w',v)  \\
=&~ \lfloor (2/5 ecc(v) + 4/7) / (1+\varepsilon) \rfloor + \dist(w',v).
\end{align*}
This allows us to bound $\dist(w',v)$ as follows:
\begin{align*}
\dist(w',v) 
&< 3/5 ecc(v) - 4/7 \\
&\hspace{20pt}- \lfloor (2/5 ecc(v) + 4/7) / (1+\varepsilon) \rfloor \\
&\le 3/5 ecc(v) - 4/7 - (2/5 ecc(v))/(1+\varepsilon) \\
&\hspace{20pt}+ \left( 1-\frac{4}{7} \cdot \frac{1}{1+\varepsilon}\right) \\
&=  3/5 ecc(v) + 3/7 - (2/5 ecc(v))/(1+\varepsilon) \\
&\hspace{20pt}- \left( \frac{4}{7} \cdot \frac{1}{1+\varepsilon}\right) \\
\end{align*}
So we now have some $w' \in W$ which satisfies:
\begin{align*}
\tilde{ecc}(v)
&\ge ecc(w) - 2(1+\varepsilon)\dist(w',v) \\
&\ge ecc(w) - 2(1+\varepsilon)\left(
\frac{3}{5} ecc(v) + 3/7 \right.\\
&\hspace{20pt}\left.- \left(\frac{2}{5} ecc(v)\right)/(1+\varepsilon) + (4/7)/(1+\varepsilon)
\right)\\
&\ge ecc(w) + (1+\varepsilon)\left(
-\frac{6}{5} ecc(v) - 6/7\right) \\
&\hspace{20pt}+ 4/5 ecc(v) + 8/7\\
&\ge \frac{3-6\varepsilon}{5} ecc(v) - (1+\varepsilon)(6/7) + 8/7\\
&\ge \frac{3-6\varepsilon}{5} ecc(v) - 4/7
\end{align*}
\else
\begin{align*}
3/5 ecc(v) - 4/7
&> \dist(w,v) 
= \dist(w,w') + \dist(w',v)  \\
&= \lfloor (2/5 ecc(v) + 4/7) / (1+\varepsilon) \rfloor + \dist(w',v).
\end{align*}
This allows us to bound $\dist(w',v)$ as follows:
\begin{align*}
\dist(w',v) 
&< 3/5 ecc(v) - 4/7 - \lfloor (2/5 ecc(v) + 4/7) / (1+\varepsilon) \rfloor \\
&\le 3/5 ecc(v) - 4/7 - (2/5 ecc(v))/(1+\varepsilon) + \left( 1-\frac{4}{7} \cdot \frac{1}{1+\varepsilon}\right) \\
&=  3/5 ecc(v) + 3/7 - (2/5 ecc(v))/(1+\varepsilon) - \left( \frac{4}{7} \cdot \frac{1}{1+\varepsilon}\right) \\
\end{align*}
So we now have some $w' \in W$ which satisfies:
\begin{align*}
\tilde{ecc}(v)
&\ge ecc(w) - 2(1+\varepsilon)\dist(w',v) \\
&\ge ecc(w) - 2(1+\varepsilon)\left(
\frac{3}{5} ecc(v) + 3/7 - \left(\frac{2}{5} ecc(v)\right)/(1+\varepsilon) + (4/7)/(1+\varepsilon)
\right)\\
&\ge ecc(w) + (1+\varepsilon)\left(
-\frac{6}{5} ecc(v) - 6/7\right) + 4/5 ecc(v) + 8/7\\
&\ge \frac{3-6\varepsilon}{5} ecc(v) - (1+\varepsilon)(6/7) + 8/7\\
&\ge \frac{3-6\varepsilon}{5} ecc(v) - 4/7
\end{align*}
\fi
For the last inequality we assume $\varepsilon \le 1$.

\end{proof}

%% file: appendix_exact_diameter.tex
\section{Dynamic Exact Diameter}
\label{app:exactDiameter}

Here we will prove the following result for maintaining the diameter exactly.

\begin{theorem}[Exact Diameter, Directed, non-trivial ($\gg n^2$)]
\label{thm:exactDiam_unweighted}
Let $G$ be a directed graph with $n$ nodes and integer weights from $\{1,2,...,W\}$.
Then for any $0 \le s, \mu \le 1$ there exists a Monte Carlo dynamic algorithm
that maintains the exact diameter of $G$ in
$\tilde{O}(
Wn^{s+\omega(1,1,\mu)-\mu}
+ Wn^{\omega(1,\mu+s,1)}
+ n^{3-s} )$ update time.
The pre-processing requires $\tilde{O}(Wn^3)$ time.

The update time is $O(n^{\updateExactDiameter})$ for current $\omega$
and $\mu \approx \muExactDiameter, s \approx \sExactDiameter$.
\footnote{
	This means for current values of $\omega$,
	this is the first non-trivial exact dynamic diameter algorithm.
	However, if fast-matrix-multiplication is further improved,
	this algorithm will be subsumed by the trivial option of just
	re-computing the diameter after each update.
}
\end{theorem}

Similar to our other results, the idea is to maintain large and small diameter separately.
We start with the case of small diameters.

\begin{theorem}[$n^s$-length distances, exact, positive integer weights]
\label{thm:dynamicSmallDiameterIntegerWeight}
Let $G$ be a graph with $n$ nodes and positive integer edge weights from $\{1,...,W\}$.
Then for any $0 \le s, \mu \le 1$ there exists a dynamic algorithm
that maintains the diameter of $G$.
If the diameter is at most $n^s$,
then the diameter is returned after each edge update,
otherwise the dynamic algorithm returns a set the pairs
$P \subset V \times V$ with $\dist(u,v) > n^s$ for $(u,v) \in P$.

The update time is
$\tilde{O}(
	n^{s+\omega(1,1,\mu)-\mu}
	+n^{\omega(1,\mu+s,1)}
)$
and the pre-processing is $\tilde{O}(\min\{Wn^{\omega+s}, Wn^3\})$.
\end{theorem}

\begin{proof}
We run a modification of \Cref{cor:maintainApproximateGraphSlices}.
We first explain how to obtain the diameter,
then we explain how to modify the algorithm to obtain a faster update complexity.

\paragraph{Maintaining the diameter}

We choose $S = \{1,2,...,n^s\}$, so for any $ 1 \le k \le n^s$ we can query for all pairs, if the distance is at most $k$ in $\tilde O(sn^{\omega(1,\mu+s,1)})$ time.
After every edge update, we query this information for $k=n^s$.
If every pair has distance at most $n^s$, then the diameter is smaller.
In that case, we simply binary search for it.
If the diameter is larger than $n^s$, then we also obtain the pairs, for which the distance is larger than $n^s$.

\paragraph{Improving the update time}

Usually the reset of \Cref{cor:maintainApproximateGraphSlices} would cost $\tilde{O}(s|S| n^{\omega(1,1,s+\mu)}) = \tilde{O}(n^{\omega(1,1,s+\mu)+s})$ time.
This is because we compute the coefficients of $X^d$ for $d \in S$ in \Cref{lem:dynamicApproximateSlice} in $\tilde{O}(n^{\omega(1,1,s+\mu)})$ operations each. 
This can be reduced to $\tilde{O}(n^{\omega(1,1,\mu)+s})$ in total for all $d \in S$ together, 
by just computing the entire inverse $(M+UV^\top)^{-1} = \hat{U} \hat{V}^\top$ at once.

By choosing $\nu = \mu$ in \Cref{lem:dynamicMatrixInverse} the update time for \Cref{cor:maintainApproximateGraphSlices} becomes
$\tilde{O}(
	n^{s+\omega(1,1,\mu)-\mu}
	+n^{1+\mu+s}$.
\footnote{Another way to obtain this algorithm would be to apply 
\Cref{lem:computeSlice} to Sankowski's algorithm \Cref{lem:dynamicMatrixInverse} via some white-box reduction. This is why the update time here is identical to \Cref{lem:dynamicMatrixInverse}.
}
As we perform some queries after every update, the update time for \Cref{thm:dynamicSmallDiameterIntegerWeight} is
$\tilde{O}(
	s n^{s+\omega(1,1,\mu)-\mu}
	+s n^{\omega(1,\mu+s,1)})$.

\paragraph{Pre-processing}

If $s$ is large, then the pre-processing of \Cref{lem:dynamicMatrixInverse} and \Cref{lem:dynamicApproximateSlice} become slower than $O(Wn^3)$. The pre-processing of both algorithms consists of computing the inverse of a polynomial matrix. 
One can compute the inverse of a polynomial matrix of degree $W$ in $\tilde{O}(Wn^3)$ operations \cite{ZhouLS15}, which can be faster than $O(n^{\omega+s})$ for large enough $s$.
(This is not useful for speeding up \Cref{lem:dynamicApproximateSlice} in general, 
as the matrix $M$ could have degree upto $n^s \gg W$. 
We can only use it here, because we assume the edge weights are bounded by $W$.)

\end{proof}

\begin{theorem}\label{thm:exactDistanceLargeHop}

Let $G$ be a directed graph with $n$ nodes and non-negative edge weights
and let $P \subset V \times V$ be a set of pairs $u,v$,
such that their shortest path uses at least $d$ hops.

Then for any such $d$ we can compute with high probability
$\dist_G(u,v)$ for every pair $(u,v) \in P$ in $O(n^3/d)$ time.
\end{theorem}

\begin{proof}

We first Sample $\tilde{O}(n/d)$ random nodes and call this set $H \subset V$.
Via \Cref{lem:allPathsHittingSet} we know that (w.h.p)
every shortest path, using at least $d$ hops, must visit at least one node from $H$.

We now compute the shortest-path-trees for every node in $H$ (in both directions).
So when $D$ is the distance matrix of $G$,
then we now know $D_{H,V}$ and $D_{V,H}$.
Further, we have $D_{u,v} = D_{u,H} \star D_{H,v}$ w.h.p for every pair $(u,v) \in P$,
because their paths use at least $d$ hops.

The total time for all shortest-path-trees and the $(min,+)$-product
$D_{u,H} \star D_{H,v}$ is at most $O(n^3/d)$ time.
\end{proof}

\begin{proof}[Proof of \Cref{thm:exactDiam_unweighted}]
We maintain the diameter (when bounded by $Wn^s$)
via \Cref{thm:dynamicSmallDiameterIntegerWeight}.

If the diameter is larger than $Wn^s$, then we obtain a set of pairs
$P \subset V \times V$ whose distance is greater than $Wn^s$,
so the shortest paths connecting these pairs use at least $n^s$ hops.
We now use \Cref{thm:exactDistanceLargeHop} to compute the distances
between these pairs in $\tilde{O}(n^{3-s})$ time.
The longest of these distances is the diameter of the graph.

The update time of the algorithm is thus $\tilde{O}(
Wn^{s+\omega(1,1,\mu)-\mu}
+W n^{\omega(1,\mu+s,1)}
+ n^{3-s} )$.

Note that the algorithm works against non-oblivious adversary,
because the output is a single number which is correct with high probability.
\end{proof}

%% file: appendix_closeness_centrality.tex
\section{Dynamic Closeness-Centrality}
\label{app:closenessCentrality}

The closeness-centrality $c(s)$ of a node $s$ is defined to be
$$c(s) := (n-1) / \sum_{v \in V} \dist(v,s),$$
which is just the inverse of the average distances to $s$.

\begin{theorem}
\label{thm:dynamicClosenessCentrality}

Let $G$ be an undirected, unweighted graph with $n$ nodes.
Then there exists a Monte Carlo dynamic algorithm with $O(n^{\updateUnwUndAPSPquery} / \varepsilon^{\omega+1})$ update time, that maintains $(1+\varepsilon)$-approximate closeness-centrality for all nodes.
The pre-processing time is $O(n^{\preUnwUndAPSPquery})$.

(For current $\omega$, the pre-processing and update time 
is the same as in \Cref{thm:undirectedApproximateDistancesQuery})

\end{theorem}

The high-level idea of \Cref{thm:dynamicClosenessCentrality} is as follows:
Intuitively it should be possible to approximate this average by
sampling $k$ nodes $v_1,...,v_k$,
computing single-source distances for these nodes,
and lastly computing $\sum_{i=1}^k \dist(v_i,s) n / (k (n-1))$,
because the expectation matches $1/c(s)$.

The approximate distances for these $k$ sources can be dynamically maintained using our dynamic distance query algorithm \Cref{thm:undirectedApproximateDistancesQuery}.

We will prove that this works, if we choose
$k = \tilde{\Theta}(n^{2/3} / \varepsilon^2)$ samples,
in which case we obtain a $(1+\varepsilon)$ approximation
of the average with high probability.
The inverse of that value is then also a
$(1+\varepsilon)$ approximation of the closeness-centrality.
This is an improvement over \cite{EppsteinW04} for large diameter graphs,
as the algorithm in \cite{EppsteinW04} required $k = \Omega(\Delta^2 / \varepsilon^2)$ samples,
where $\Delta$ is the diameter.
Using our result, one can obtain a static combinatorial 
$\tilde O(mn^{2/3} / \varepsilon^2)$ time algorithm by running BFS from the $k$ sources. 
This would be faster than \cite{EppsteinW04} once $\diam(G) > n^{1/3}$.

\begin{lemma}\label{lem:closenessSamplingSize}
Let $G$ be an unweighted $n$-node graph, $k = \Theta(n^{2/3} / \varepsilon^2)$ and $v_1,...,v_k$ be a uniformly chosen sample of $k$ nodes.

Then for any $s \in V$ the term $\sum_{i=1}^k \dist(v_i,s) n / (k (n-1))$ is a $(1+\varepsilon)$-approximation of $1/c(s) = \sum_{v \in V} \dist(v,s) / (n-1)$ with high probability (assuming $1/c(s)$ is finite).
\end{lemma}

\begin{proof}[Proof of \Cref{thm:dynamicClosenessCentrality}]
We run the dynamic algorithm of \Cref{thm:undirectedApproximateDistancesQuery}.
After every update we simply sample 
$k = \tilde{\Theta}(n^{2/3} / \varepsilon^2)$ random nodes $v_1,...,v_k$
as in \Cref{lem:closenessSamplingSize}.
Then $(\sum_{i=1}^k \dist(v_i,s) n / (k (n-1)))^{-1}$ will be a $(1+\varepsilon)$-approximation,
if $c(s)$ is non-zero.
The case that $c(s)$ is zero (i.e. there exists a pair $u,v \in V$ with $\dist(u,v) = \infty$) can be handled by running \Cref{lem:strongConnectivity} in parallel.
If we detect via \Cref{lem:strongConnectivity} that the graph is not connected, 
then we simply return $0$ as the closeness centrality.

When querying the distances for $\tilde O(n^{2/3} / \varepsilon^2) \times n$ pairs via \Cref{thm:undirectedApproximateDistancesQuery} (for the same parameter $s$ and $\mu$),
then the query time is 
$\tilde{O}(n^{\omega(2/3, s+\mu, 1)} / \varepsilon^2) 
= \tilde{O}(n^{\omega(2/3, 0.45, 1)} / \varepsilon^2) 
= \tilde O(n^{1.76} / \varepsilon^3)$, 
which is subsumed by the update time.
\end{proof}

An important tool for proof \Cref{lem:closenessSamplingSize} is the Hoeffding-inequality:
\begin{lemma}[Hoeffdings-inequality]
Let $X_1,...,X_N$ be independent random variables bounded by the interval $[a,b]$, then
$$
\mathbb{P}\left[
\left|
\frac{1}{n}\sum_{i=1}^N X_i
-
\mathbb{E}\left[
\sum_{i=1}^N X_i
\right]
\right|
 \ge t \right]
\le
2 \exp\left(
- \frac{2Nt^2}{(b-a)^2}
\right)
$$
\end{lemma}

Unfortunately, just applying Hoeffdings-inequality to the sampled distances is not enough.
If for some source node $s$ we sample a few nodes $v_1,...,v_k$ and compute $\dist(s,v_i)$ for $i=1,...,k$,
then $\dist(s,v_i)$ are random variables in the interval $[1,\diam(G)]$.
This causes the additive error to become quite large for large diameter graphs (hence the $\varepsilon \cdot \diam(G)$ error bound in \cite{EppsteinW04}).

We are able to fix this issue by splitting the analysis into two cases: long and short paths.
The idea is as follows:
Let $d$ be some parameter and $u \in V$ be a node with $d \le \dist(s,u) \le 2d$.
Then there exist at least $d/2$ many nodes $v \in V$ (half of the nodes along the path from $s$ to $u$) 
with $d/2 \le \dist(s,v) \le d$.
So for large enough $d$ and $k$, we have that w.h.p there are at least $\Omega(dk/n)$ many sampled nodes $v_i$ with $d/2 \le \dist(s,v_i) \le d$.
Note that for these random variables the interval $[d/2, d]$ has length $d/2$, which is proportional to the number of samples $\Omega(dk/n)$. This allows us to give better bounds on the error probability.
However, this approach only works if $d$ is large enough, 
so we also need to handle the case of small paths separately.

\begin{proof}[Proof of \Cref{lem:closenessSamplingSize}]

Fix some node $s \in V$.
Let $k_{\le d}$ be the number of sampled nodes $v_i$ with $\dist(v_i,s) \le d$
and $k_{>d}$ be the number of nodes with $\dist(v_i,s) \ge d$.
Likewise let $n_{\le d}$ and $n_{>d}$ be the number of nodes $v \in V$
with $\dist(v_i,s) \le d$ or $\dist(v_i,s) > d$ respectively.

Then we can write
\begin{align*}
\sum_{i=1}^k \dist(v_i,s) / k
= \frac{
k_{\le d} \sum_{i : \dist(v_i,s) < d}   \dist(v_i,s) / k_{\le d}
+
k_{> d}   \sum_{i : \dist(v_i,s) \ge d} \dist(v_i,s) / k_{> d}
}{ k }
\end{align*}

We will search for lower bounds on $k$ which imply that
$\sum_{i : \dist(v_i,s) < d} \dist(v_i,s) / k_{\le d}$
has an additive error of at most $\varepsilon k / (c(s) k_{\le d})$
compared to $\sum_{v : \dist(v,s) < d} \dist(v,s) / n_{\le d}$,
and that $\sum_{i : \dist(v_i,s) \ge d} \dist(v_i,s) / k_{> d}$
is a $(1+\varepsilon)$ approximation of
$\sum_{v : \dist(v,s) \ge d} \dist(v,s) / n_{> d}$.

When these conditions are met, then $\sum_{i=1}^k \dist(v_i,s) / k$
is a $(1+2\varepsilon)$ approximation of the average distance.

\paragraph{Long distances}

Let $d_{\max} := \max_{v \in V} \dist(v,s)$ and $j := \lceil \log d_{\max} /d \rceil$
and $b = \sqrt[j]{d_{\max}/d}$, which means we have $j \le \log n$, $d b^{j} = d_{\max}$
and $2 \le b \le 4$.
Further let $k_l$ be the number of sampled nodes $v_i$ such that
$d b^l \le \dist(v_i,s) < d b^{l+1}$ and likewise let $n_l$ be
the number of nodes $v \in V$ with $d b^l \le \dist(v,s) < d b^{l+1}$.

We can split the sum of distances as follows:
\begin{align*}
\sum_{i : \dist(v_i,s) > d} \dist(v_i,s)
=
\sum_{l = 0}^{j-1} \sum_{i : d b^l < \dist(v_i,s) \le d b^{l+1}} \dist(v_i,s)
=
\sum_{l = 0}^{j-1} k_l \sum_{i : d b^l < \dist(v_i,s) \le d b^{l+1}} \dist(v_i,s) / k_l
\end{align*}
Note that in the graph there exist at least
$db^{l+1} - db^l \ge d2^l$
many nodes with
$d b^l < \dist(v_i,s) \le d b^{l+1}$ for $l < j$.
Thus when sampling $k$ nodes, we expect $k_l \ge (k/n) d2^l$.
Also via Chernoff-bound we know for $k = \tilde{\Omega}(n/(d2^l) \varepsilon^{-1})$
that we have $k_l = \Omega((k/n) d2^l)$ with high probability.

Let $\xi$ be a bound on the additive error of
$\sum_{i : d b^l < \dist(v_i,s) \le d b^{l+1}} \dist(v_i,s) / k_l$
compared to the expectation
$\mu = \sum_{v : d b^l < \dist(v,s) \le d b^{l+1}} \dist(v,s) / n_l$.
Then we want this bound to be $\xi = db^{l}\varepsilon$, because
with this additive error bound we obtain a $(1+\varepsilon)$ multiplicative error.
If this holds for all $l$, then this implies that 
$\sum_{i : \dist(v_i,s) > d} \dist(v_i,s) / k_{>d}$
is a $(1+\varepsilon)$ approximation of
$\sum_{v : \dist(v,s) > d} \dist(v_i,s) / n_{>d}$.

The probability of having a larger additive error can be bounded by Hoeffdings-inequality:

\begin{align*}
\mathbb{P}\left[ \left| \frac{1}{k_l}\sum_{i : d b^l < \dist(v_i,s) \le d b^{l+1}} \dist(v_i,s) - \mu \right| > \xi \right]
& \le 2 \exp(-2k_l\xi^2 / (d2^l)^2)
= \exp -\tilde{\Omega} (k/(nd2^l) \xi^2)\\
& = \exp -\tilde{\Omega} (db^l k \varepsilon^2 /n )
= \exp -\tilde{\Omega} (d k \varepsilon^2 /n )
\end{align*}

So to bound the probability by some $n^{-c}$ for constant $c$,
it suffices to choose $k = \tilde{\Omega}(n/d \varepsilon^{-2})$.

\paragraph{Short distances}

For short distances, we want the distance threshold $d$
to be at least the median distance to $s$. We will write $\operatorname{med}(s)$ for the median distance from $s$, i.e. the median of $\dist(s,v)$ for $v \in V$.
For now assume $d = h \cdot \operatorname{med}(s)$ for some $h \ge 1$.
This means that $k_{\le d} \ge 0.5 k$ with high probability for $k = \Omega(\log n)$.
We want to bound the additive error of
$\sum_{i : \dist(v_i,s) < d} \dist(v_i,s) / k_{\le d}$
compared to
$\mu = \sum_{v : \dist(v,s) < d} \dist(v_i,s) / n_{\le d}$
to be at most
$\xi = \frac{\varepsilon k}{c(s) k_{\le d}}$.

Note that $1/c(s) \ge \mu \ge 0.5 n_{\le d} \cdot \operatorname{med}(s)$, so $d \le 2h\mu \le 2h/c(s)$.
The probability that we have a larger additive error can thus be bounded by
\begin{align*}
\exp(-2 k_{\le d} \xi^2 / d^2)
\le
\exp\left( -2 k_{\le d}
	\left(\frac{  \varepsilon k  }{  c(s) k_{\le d}  }\right)^2 /
	\left(\frac{  2h  }{  c(s)  }\right)^2
\right)
\le
\exp(-\Omega\left(k \varepsilon^2 / h^2\right))
\end{align*}
So for $d = h \cdot \operatorname{med}(s)$, $h \ge 1$ and $k = \tilde{\Omega}(h^2 / \varepsilon^2)$, the error does not become too large with high probability.

\paragraph{Balancing $k$}

We now know that $k$ must satisfy $k = \tilde{\Omega}(n/d \varepsilon^{-2})$,
and $k = \tilde{\Omega}(h^2 / \varepsilon^2)$,
where $d = h \cdot \operatorname{med}(s)$ and $h \ge 1$.
We claim that choosing $k = \tilde{\Theta}(n^{2/3} \varepsilon^{-2})$
is enough to satisfy all these conditions:

\begin{itemize}

\item If the median $\operatorname{med}(s)$ is less than $n^{1/3}$,
then we choose $d = n^{1/3}$
so $k = \tilde{\Theta}(n^{2/3} \varepsilon^{-2})$ satisfies all required condition.

\item If the median $\operatorname{med}(s)$ is more than $n^{1/3}$,
then we can choose $h$ to be some constant.
This results in $d = \Omega(n^{1/3})$ and thus the conditions
$k = \tilde{\Omega}(n/d \varepsilon^{-1})$ and $k = \Omega(h^2 / \varepsilon^2)$
are both satisfied when choosing $k = \tilde{\Theta}(n^{2/3} \varepsilon^{-2})$.

\end{itemize}

Via union bound \Cref{lem:closenessSamplingSize} holds for all nodes $s \in V$, if the constant hidden in $k = \tilde{\Theta}(n^{2/3} \varepsilon^{-2})$ is large enough.

\end{proof}